\documentclass{paper}

\usepackage{geometry}
\usepackage{amsmath,amsthm,amsfonts,amssymb,latexsym,graphicx}
\usepackage{geometry,xcolor,hyperref}
\usepackage{todonotes}
\usepackage{url}
\usepackage{paralist}
\usepackage{cases}
\usepackage[noend]{algorithmic}
\usepackage[boxed]{algorithm}
\usepackage{paralist}
\usepackage{color,xcolor,hyperref}
\usepackage{footmisc}
\usepackage{caption}
\usepackage{mathtools}

\usepackage{multirow}

\usepackage{subfigure}
\usepackage{graphicx,xspace}
\newtheorem{theorem}{Theorem}[section]
\newtheorem{ques}{Question}
\newtheorem{lemma}[theorem]{Lemma}

\newtheorem{cor}[theorem]{Corollary}
\newtheorem{definition}[theorem]{Definition}

\newtheorem{corollary}[theorem]{Corollary}

\newcommand{\quasipoly}{\ensuremath{\mathrm{quasipoly}}}
\newcommand{\poly}{\ensuremath{\mathrm{poly}}}
\newcommand{\polylog}{\ensuremath{\mathrm{polylog}}}

\newcommand{\sepsa}[1]{\ensuremath{\cC\cS(#1)}}
\newcommand{\sepse}[1]{\ensuremath{\cC\cS_{\text{\fontsize{1}{1}\selectfont\hspace{-0.05em}\scalebox{0.6}{$=$}\hspace{-0.15em}}}(#1)}}
\newcommand{\mult}{\ensuremath{\mu}}

\newcommand{\tw}{\ensuremath{\text{tw}}}

\newcommand{\cA}{\ensuremath{\mathcal{A}}}
\newcommand{\cB}{\ensuremath{\mathcal{B}}}
\newcommand{\cC}{\ensuremath{\mathcal{C}}}

\newcommand{\cF}{\ensuremath{\mathcal{F}}}

\newcommand{\cI}{\ensuremath{\mathcal{I}}}

\newcommand{\cM}{\ensuremath{\mathcal{M}}}
\newcommand{\cO}{\ensuremath{\mathcal{O}}}

\newcommand{\cS}{\ensuremath{\mathcal{S}}}

\newcommand{\into}{\hookrightarrow}

\newcommand{\dist}{\ensuremath{\mathrm{dist}}}

\newcommand{\ans}[2]{#1[#2]}

\DeclareRobustCommand{\extends}{\operatorname{extends}}
\DeclareRobustCommand{\picks}{\operatorname{picks}}

\newcommand{\algcomment}[1]{\colorbox{black!10}{#1}}

\newcommand{\LineIfElse}[3]{ \STATE \algorithmicif\ {#1}\ \algorithmicthen\ {#2} \algorithmicelse\ {#3} }

\newcommand{\priv}{\ensuremath{\mathsf{Priv}}}
\newcommand{\pub}{\ensuremath{\mathsf{Pub}}}

\newcommand{\auto}{\ensuremath{\mathrm{auto}}}

\newcommand{\treedecomp}{\ensuremath{\mathbb{T}}\xspace}
\newcommand{\pointleft}{\leftarrow}
\newcommand{\pointdown}{\downarrow}
\newcommand{\pointright}{\rightarrow}
\newcommand{\pointup}{\uparrow}

\newcommand{\ann}{\mathtt{ann}}

\newcommand{\reach}{\ensuremath{\mathrm{reach}}}

\renewcommand{\int}{\ensuremath{\mathrm{int}}}
\newcommand{\ext}{\ensuremath{\mathrm{ext}}}
\newcommand{\sint}{\ensuremath{\mathrm{sint}}}
\newcommand{\sext}{\ensuremath{\mathrm{sext}}}

\newcommand{\side}{\ensuremath{\mathrm{side}}}

\newcommand{\sub}{\ensuremath{\mathrm{sub}}}
\newcommand{\ind}{\ensuremath{\mathrm{ind}}}

\newcommand{\heavy}{\mathsf{he}}
\newcommand{\light}{\mathsf{li}}
\newcommand{\discard}{\mathsf{di}}

\newcommand{\low}{\ensuremath{\mathrm{low}}}
\newcommand{\upp}{\ensuremath{\mathrm{upp}}}

\newcommand{\smallM}{\ensuremath{\mathrm{small}}}
\newcommand{\largeM}{\ensuremath{\mathrm{large}}}

\newcommand{\can}{\ensuremath{\mathrm{can}}}

\newcommand{\feas}{\ensuremath{\mathrm{feas}}}

\newcommand{\boundary}[1]{\ensuremath{\partial{#1}}}

\newcommand{\tO}{\tilde{O}}

\title{\centering Detecting and Counting Small Patterns in Planar Graphs \\ in Subexponential Parameterized Time}
\author{\hspace{12em} Jesper Nederlof\thanks{Eindhoven University of Technology. \href{j.nederlof@tue.nl}{j.nederlof@tue.nl}. Supported by the Netherlands Organization for Scientific Research under project no. 024.002.003 and the European Research Council under project no. 617951}}

\begin{document}
\maketitle
\setcounter{page}{0}
\thispagestyle{empty}
\begin{abstract}
We present an algorithm that takes as input an $n$-vertex planar graph $G$ and a $k$-vertex pattern graph $P$, and computes the number of (induced) copies of $P$ in $G$ in $2^{O(k/\log k)}n^{O(1)}$ time. If $P$ is a matching, independent set, or connected bounded maximum degree graph, the runtime reduces to $2^{\tilde{O}(\sqrt{k})}n^{O(1)}$.

While our algorithm \emph{counts} all copies of $P$, it also improves the fastest algorithms that only \emph{detect} copies of $P$.
Before our work, no $2^{O(k/\log k)}n^{O(1)}$ time algorithms for detecting unrestricted patterns $P$ were known, and by a result of Bodlaender et al.~[ICALP 2016] a $2^{o(k/\log k)}n^{O(1)}$ time algorithm would violate the Exponential Time Hypothesis (ETH).
Furthermore, it was only known how to detect copies of a fixed connected bounded maximum degree pattern $P$ in $2^{\tilde{O}(\sqrt{k})}n^{O(1)}$ time \emph{probabilistically}.

For counting problems, it was a repeatedly asked open question whether $2^{o(k)}n^{O(1)}$ time algorithms exist that count even special patterns such as independent sets, matchings and paths in planar graphs.
The above results resolve this question in a strong sense by giving algorithms for counting versions of problems with running times equal to the ETH lower bounds for their \emph{decision} versions.

Generally speaking, our algorithm counts copies of $P$ in time proportional to its number of non-isomorphic separations of order $\tilde{O}(\sqrt{k})$.
This algorithm introduces a new recursive approach to construct families of balanced cycle separators in planar graphs that have limited overlap inspired by methods from Fomin et al.~[FOCS 2016], a new `efficient' inclusion-exclusion based argument and uses methods from Bodlaender et al.~[ICALP 2016].
\end{abstract}

\newpage

\section{Introduction}
The complexity of NP-hard problems on planar graphs has been a popular subject for a at least two decades, and its fruitful study resulted in seminal results such as the planar separator theorem by Lipton and Tarjan~\cite{doi:10.1137/0136016} and efficient approximation schemes by Baker~\cite{DBLP:journals/jacm/Baker94}.
An area in which planar graphs are especially a popular subject of study is \emph{Parameterized Complexity}.
A cornerstone result of parameterized complexity\footnote{Quoting its laudatio for the Myhil-Nerode prize~\cite{laudatio}.} by Fomin et al.~\cite{DBLP:journals/jacm/DemaineFHT05} shows that many NP-hard parameterized problems on planar graphs can be solved in \emph{subexponential time}, i.e. $f(k)n^{O(1)}$ time where $f$ is $2^{o(k)}$, $k$ is some (typically small) problem parameter, and $n$ denotes the number of vertices of $G$. Typically $f(k)$ is only $2^{\tilde{O}(\sqrt{k})}$ in this setting.
Thanks to the technique of~\cite{DBLP:journals/jacm/DemaineFHT05} and a large body of follow-up work (see e.g.~\cite{erickson_et_al:DR:2016:6722}), the fine-grained parameterized complexity of many decision problems on planar graphs is by now well understood.

The technique from~\cite{DBLP:journals/jacm/DemaineFHT05}, called \emph{bidimensionality}, is a win-win argument based on the \emph{grid minor theorem}.
The technique exploits that instances of the problem at hand defined by graphs with high treewidth are always YES/NO-instances.\footnote{See Section~\ref{sec:prel} for the definition of treewidth and~\cite[Section 7.7]{Cygan:2015:PA:2815661} for more discussion.}
For example, one can detect whether a planar graph $G$ has a simple path on at least $k$ vertices (called $k$-path) in time $2^{O(\sqrt{k})}n^{O(1)}$ time in this way: If $G$ has treewidth $\Omega(\sqrt{k})$, it has a $(\Omega(k)\times\Omega(k))$-grid as a minor that can be used to show $G$ has a $k$-path. Otherwise, $G$ has treewidth $O(\sqrt{k})$ and dynamic programming can be used to detect $k$-paths in $2^{O(\sqrt{k})}n^{O(1)}$ time.
While the bidimensionality technique is applicable to many problems, it requires that the solution of the instance can be deduced already from the fact that the graph has large treewidth.
This is a rather fragile assumption that often can not be made, and indeed for several important problems the approach turned out inadequate.

\paragraph{Subgraph Isomorphism}\hspace{-1em} on planar graphs is a basic NP-complete problem where the bidimensionality technique falls short perhaps most pressingly.
In the subgraph isomorphism problem we are given an $n$-vertex planar graph $G$ and a $k$-vertex planar graph $P$ and we need to determine whether there exists an (induced) copy of $P$ in $G$.
As mentioned above, bidimensionality does solve this problem in $2^{O(\sqrt{k})}n^{O(1)}$ if $P$ is a path on $k$ vertices, but if we slightly alter the pattern to, say, a cycle on $k$ vertices or a directed path\footnote{The subgraph isomorphism problem can be extended to directed graphs in a natural way.} the technique already breaks down.

Detecting such cycles patterns and directed paths turns out more complicated.
It was observed by Tazari~\cite{DBLP:journals/tcs/Tazari12} and Dorn et al. \cite{DBLP:journals/iandc/DornFLRS13} that variants of the layering technique by Baker~\cite{DBLP:journals/jacm/Baker94} can be used to design algorithms for detecting such patterns with running time $2^{\varepsilon k}n^{O(1/\varepsilon)}$.
Recently it was shown how to detect such patterns \emph{probabilistically} in $2^{\tilde{O}(\sqrt{k})}n^{O(1)}$ time by Fomin et al.~\cite{pattern}.
The following question suggests itself:

\begin{ques}\label{q1}
	Is there a $2^{\tilde{O}(\sqrt{k})}n^{O(1)}$ time \emph{deterministic} algorithm for Subgraph Isomorphism for the special case where the pattern is either a cycle or a directed path on $k$ vertices?
\end{ques}

The general planar subgraph isomorphism problem with unrestricted patterns also been the subject of several interesting works.
Eppstein~\cite{DBLP:journals/jgaa/Eppstein99} was the first to show that the problem is Fixed Parameter Tractable by giving an $k^{O(k)}n^{O(1)}$ time algorithm, and Dorn~\cite{DBLP:conf/stacs/Dorn10} improved this to an $2^{O(k)}n^{O(1)}$ time algorithm.
Afterwards Bodlaender et al.~\cite{DBLP:conf/icalp/BodlaenderNZ16} showed the problem can be solved in $2^{O(n / \log n)}$ time, and any $2^{o(n / \log n)}$ time would contradict the exponential time hypothesis.

By combining the techniques of Bodlaender et al.~\cite{DBLP:conf/icalp/BodlaenderNZ16} and Fomin et al.~\cite{pattern}, one can obtain a $2^{O(k/\log k)}n^{O(1)}$ time probabilistic algorithm that detects any fixed pattern with at most $O(\sqrt{k}/\log k)$ connected components.
However, this still does not settle the complexity of the general problem, and the following question (also mentioned in~\cite{pattern} and~\cite{DBLP:conf/icalp/BodlaenderNZ16}) remained open:
\begin{ques}\label{q2}
	Is there a $2^{O(k/\log k)}n^{O(1)}$ time (deterministic) algorithm for general Subgraph Isomorphism with unrestricted pattern?
\end{ques}
Note that such an algorithm can not be improved under the ETH by the lower bound $2^{o(n / \log n)}$ lower bound from~\cite{DBLP:conf/icalp/BodlaenderNZ16} since $k \leq n$ in any non-trivial instance.

\paragraph{Counting Problems}\hspace{-1em} are perhaps the largest category of problems for which bidimensionality is not applicable.
Counting problems on restricted graphs classes are well-motivated from (among others) seemingly distant areas such as statistical physics, and 
Counting problems on planar graphs are well-studied in terms of polynomial time approximation schemes (see e.g Goldberg~\cite{DBLP:journals/jcss/GoldbergJM15}), and several works showed that the study of 
approximation schemes and parameterized complexity of counting problems is intertwined: Both the methods from Yin~\cite{Yin:2013:ACV:2627817.2627821} and Patel and Regts~\cite{patelregts} give polynomial time approximation schemes that rely on completely unrelated fixed parameter tractable algorithms.

On the other hand, the number of purely parameterized complexity theoretical works on counting problems can be counted on one hand. For example, Frick~\cite{DBLP:journals/mst/Frick04} gave a fixed parameter tractable algorithm for a general class of fixed order logic problems, and Curticapean~\cite{DBLP:conf/esa/Curticapean16} showed it is fixed parameter tractable to count matchings with few unmatched vertices in planar graphs.

\paragraph{Counting Subgraph Isomorphisms}\hspace{-1em} is a very natural extension of subgraph isomorphism with close connections to partition functions~\cite{patelregts} and motif discovery (see e.g.~\cite{Milo824} or the discussion in~\cite{DBLP:conf/stoc/CurticapeanDM17}). 
The aforementioned algorithm of Dorn~\cite{DBLP:conf/stacs/Dorn10} also counts the number of copies of the given pattern, so patterns like independent sets, matchings and paths on $k$ vertices can be counted in $2^{O(k)}n^{O(1)}$ time.
Yet, this does not match the typical $2^{o(\sqrt{k})}n^{O(1)}$ running time that can be obtained for most decision problems via the bidimensionality technique and cannot be improved under the ETH by standard reductions.
Therefore the following natural question was repeatedly asked by several researchers:\footnote{For example, it was posed as open problem in a Dagstuhl report by Marx in~\cite{cygan_et_al:DR:2017:7247}, and talks by Fomin~\url{https://ims.nus.edu.sg/events/2017/asp/files/fedor.pdf} and Saurabh~\url{https://rapctelaviv.weebly.com/uploads/1/0/5/3/105379375/future.pdf}.}

\begin{ques}\label{q3}
	Is there a $2^{\tO(\sqrt{k})}n^{O(1)}$ time algorithm for counting (induced) copies of a pattern $P$, where $P$ is a path, matching, set of disjoint triangles or independent set\footnote{Naturally, if $P$ is an independent set, only counting induced copies of $P$ is interesting.}?
\end{ques}

We would like to stress that even for the special case of counting independent sets on $k$ vertices in subgraphs of grids it is unclear how to obtain an algorithm with $2^{o(k)}n^{O(1)}$ running time without using our techniques. Note that this specific counting problems on subgraphs of grids received attention already in the enumerative combinatorics community (see~\cite{Calkin:1998:NIS:288856.288860, KASTELEYN19611209}).

\renewcommand{\arraystretch}{1.35}
\newcommand{\chec}{\checkmark}
\newcommand{\cross}{$\times$}
\begin{table}
	\centering
	\begin{tabular}{|l |p{4cm} |c |c |l|}
		\hline
		\textbf{Reference} 	& \textbf{Pattern Restriction} & \textbf{Deterministic}	& \textbf{Counting} & \textbf{Runtime}\\
		\hline
		\cite{DBLP:journals/dm/MatousekT92} 	& connected bounded degree		& \chec			& \chec				& $n^{O(\sqrt{n})}$\\
		\hline
		\cite{DBLP:journals/jgaa/Eppstein99} 	& -	  							& \chec			& \chec			& $k^{O(k)}n^{O(1)}$\\
		\hline
		\cite{DBLP:conf/wg/Dorn07} 	& undirected path	  				& \chec			& \cross			& $2^{O(\sqrt{k})}n^{O(1)}$\\
		\hline
		\cite{DBLP:conf/stacs/Dorn10}		 	& -	  							& \chec			& \chec				& $2^{O(k)}n^{O(1)}$\\
		\hline
		\cite{DBLP:conf/icalp/BodlaenderNZ16}& -	 	 						& \chec		& \chec			& $2^{O(n/ \log n)}$\\
		\hline
		\multirow[c]{2}{*}{\cite{pattern}}& connected bounded degree 	& \cross		& \cross			& $2^{\tO(\sqrt{k})}n^{O(1)}$\\\cline{2-5}
		& directed path					& \cross			& \cross				& $2^{\tO(\sqrt{k})}n^{O(1)}$\\
		\hline
		\cite{DBLP:conf/icalp/BodlaenderNZ16,pattern}	& connected 	 	 						& \cross		& \cross			& $2^{O(k/ \log k)}n^{O(1)}$\\
		\hline
		\hline
		\multirow[c]{3}{*}{This paper}										& -	 	 						& \chec		& \chec			& $2^{O(k/ \log k)}n^{O(1)}$\\\cline{2-5}
		& connected bounded degree	 	 						& \chec		& \chec			& $2^{\tO(\sqrt{k})}n^{O(1)}$\\\cline{2-5}
		& independent set, matching			& \chec		& \chec			& $2^{\tO(\sqrt{k})}n^{O(1)}$\\\cline{2-5}
		\hline
	\end{tabular}
	\caption{Runtimes of algorithms for planar subgraph isomorphism.}
	\label{tab:results}
\end{table}
\renewcommand{\arraystretch}{1}

\subsection{Our Results}
We resolve the complexity of the decision and counting variants of Subgraph Isomorphism on planar graphs, and answer the above Questions~\ref{q1}, \ref{q2} and \ref{q3} affirmatively in a strong sense.

We now state our main result for counting subgraph isomorphisms.
Formally, if $G$ and $P$ are undirected graphs, we denote $\sub(P,G)$ for the set of injective functions $f: V(P) \rightarrow V(G)$ such that $\{f(v),f(w)\} \in E(G)$ for every $\{v,w\} \in E(P)$.
Similarly, we denote $\ind(P,G)$ for the set of injective functions $f: V(P) \rightarrow V(G)$ such that $\{f(v),f(w)\} \in E(G)$ if and only if $\{v,w\} \in E(P)$, for every distinct $v,w \in V(P)$.
The running time of our algorithm depends on a pattern-specific parameter $\sigma(P)$ that we define below. 

\begin{theorem}[Main Theorem]\label{thm:main}
	There is an algorithm that takes as input a $k$-vertex graph $P$ and an $n$-vertex planar graph $G$, and outputs $|\sub(P,G)|$ and $|\ind(P,G)|$ in $2^{\tO(\sqrt{k})}(\sigma(P)n)^{O(1)}$ time.
\end{theorem}

We can also count the number of vertex subsets inducing the sought copies by dividing $|\ind(P,P)|$ or $|\sub(P,P)|$. The parameter $\sigma(P)$ is the number of non-isomorphic separations of $P$ of order $\tilde{O}(\sqrt{k})$.
Formally speaking, a \emph{separation of $P$} is a pair of vertex subsets $X,Y \subseteq V(G)$ such that $X \cup Y = V$ and there are no edges in $P$ between vertices from $X \setminus Y$ and $Y \setminus X$, and the order of $(X,Y)$ is $|X \cap Y|$.
Separations $(X,Y)$ and $(X',Y')$ are isomorphic if there is an isomorphism $f$ of $P$ such that $f(X)=X'$, and $f(v)=v$ for every $v \in X \cap Y$.

This factor in our running time is a direct consequence of reverse-engineering the technique of Bodlaender et al~\cite{DBLP:conf/icalp/BodlaenderNZ16}.
In fact, it follows from the analysis of~\cite{DBLP:conf/icalp/BodlaenderNZ16} that $\sigma(P)$ is $2^{O(k / \log k)}$ for any pattern $P$ (we spell this out in Lemma~\ref{lem:boundsigmaBod}).
Thus, we obtain the following consequence of Theorem~\ref{thm:main} that answers Question~\ref{q2} positively:
\begin{cor}
	There is an algorithm that takes as input an $k$-vertex graph $P$, and an $n$-vertex planar graph $G$ and outputs $|\sub(P,G)|$ and $|\ind(P,G)|$ in $2^{O(k / \log k)}n^{O(1)}$ time.
\end{cor}
Recall that Bodlaender et al. showed that a $2^{o(n/\log n)}$ time algorithm contradicts ETH, and thus our algorithm can probably not be improved significantly easily.

We continue our discussion with specific choices for the pattern $P$.
If $P$ is a matching, set of disjoint triangles, independent set or any connected graph with bounded maximum degree on $k$ vertices, it can be shown\footnote{See Lemma~\ref{lem:boundsigmacon}} that $\sigma(P)$ is at most $2^{\tilde{O}(\sqrt{k})}$.
Thus, Theorem~\ref{thm:main} resolves Question~\ref{q3} in the following sense:
\begin{cor}\label{cor:simplepat}
	Let $P$ be a matching, set of disjoint triangles, independent set or any connected graph with bounded maximum degree on $k$ vertices.
	Then there is an algorithm that takes as input an $n$-vertex planar graph $G$ and outputs $|\sub(P,G)|$ and $|\ind(P,G)|$ in $2^{\tilde{O}(\sqrt{k})}n^{O(1)}$ time.
\end{cor}
It is folklore knowledge that, assuming ETH, there is no $2^{o(\sqrt{n})}$ time algorithm that decides whether $|\sub(P,G)|>0$ when $P$ is a set of triangles or a path, or whether $|\ind(P,G)| > 0$ when $P$ is a matching, set of disjoint triangles or independent set or path.
Thus, our result resolves the running time of the fastest pattern counting algorithm assuming ETH for all listed pattern classes.

The algorithm for detecting a connected pattern with bounded maximum degree from Corollary~\ref{cor:simplepat} can be combined with a simple gadget\footnote{Replace each arc $(v,w)$ in the host/pattern graph with new vertices $\{a,b,c\}$ and edges $\{\{v,a\},\{a,b\},\{a,c\},\{c,w\}\}$.} to obtain the following result that resolves Question~\ref{q1}.
\begin{cor}
	There is an $2^{\tilde{O}(\sqrt{k})}n^{O(1)}$ time deterministic algorithm that detects (and in fact, even counts) the number of simple directed cycles or paths on $k$ vertices in an $n$-vertex directed planar graph.
\end{cor}

Finally, we would like to mention that, via standard techniques, our algorithms can also be used to obtain uniform samples from the set $\ind(P,G)$ and $\sub(P,G)$ in similar running times. Such questions were also studied for $P$ being a path on $k$ vertices and $G$ being planar by Montanari~\cite{DBLP:conf/mfcs/MontanariP15}.

\subsection{Previous Related Work}

A seminal paper by Alon et al.~\cite{DBLP:journals/jacm/AlonYZ95} gave a $2^{O(k)}n^{\tw(P)}$ time algorithm for subgraph isomorphism in general graphs.
For the counting extension, Flum and Grohe showed there is no $f(k)n^{O(1)}$ time algorithm that counts occurrences of $P$ in $G$ even in the special case that $P$ is a path on $k$ vertices.

On the other hand Patel and Regts~\cite{DBLP:journals/corr/PatelR17} gave an algorithm that counts the number of induced copies of $P$ in time $O(c(\Delta)^kn)$.
Curticapean et al.~\cite{DBLP:conf/stoc/CurticapeanDM17} gave an algorithm that counts the number of copies of $P$ in $G$ in $2^{O(|E(P)|\log |E(P)|)} n^{0.174|E(P)|}$ time.
See also a survey by Curticapean~\cite{DBLP:conf/iwpec/Curticapean18} on counting and parameterized complexity.

The special case of $G$ (and thus, also $P$) being \emph{planar} was studied first by Eppstein~\cite{DBLP:journals/jgaa/Eppstein99}.
His approach was to follow the layering technique of Baker~\cite{DBLP:journals/jacm/Baker94}.
Briefly speaking, this is to partition the vertex set into $k+1$ parts $V_1,\ldots,V_{k+1}$ such that for every $i$, the graph $G[V \setminus V_i]$ has treewidth $O(k)$.
Then one can try to count all occurrences of $P$ by summing over all $i$, and count all occurrences of $P$ in $G[V \setminus V_i]$ using dynamic programming on the treedecomposition.
However, this overcounts occurrences of $P$ that are disjoint from more than one part.
To avoid this overcounting, we use additional table indices to keep track of whether vertices from some part where included, and only count pattern occurrences where some fixed part $V_i$ is disjoint from the pattern occurrence, but where all $V_j$ with $j < i$ intersect with the pattern occurrence. In this way Eppstein obtained an algorithm that counts the number of occurrences of $P$ in $k^{O(k)}n$ time.
Dorn~\cite{DBLP:conf/stacs/Dorn10} later sharpened the running time to $2^{O(k)}n$ by exploiting planarity in the dynamic programming subroutine.

Bodlaender et al.~\cite{DBLP:conf/icalp/BodlaenderNZ16} settled the complexity of subgraph isomorphism with large patterns in a curious way: They showed that occurrences of $P$ can be detected in $2^{O(n / \log n)}$ time, and that any $2^{o(n/ \log n)}$ time algorithm would contradict ETH.
Their algorithm builds on a natural dynamic programming algorithm that is indexed by separations of order $O(\sqrt{n})$, but exploits that many table entries computed by this algorithm will be equal whenever the associated separations are isomorphic.
The curious running time follows from an upper bound on the number of non-isomorphic separations of order $\tilde{O}(\sqrt{n})$.

Fomin et al.~\cite{pattern} provided a new robust tool: given a planar graph\footnote{The result of Fomin et al.~\cite{pattern} applies to the more general class of apex-minor free graphs, but we restrict our discussion to planar graphs.} $G$ and an integer $k$, they sample a subset $S \subseteq V(G)$ such that $\tw(G[S])= \tO(\sqrt{k})$ and for every $X \subseteq V(G)$ such that $G[X]$ has $O(\sqrt{k}/\log k)$ connected components it holds that $X \subseteq S$ with probability at least $1/2^{\tO(\sqrt{k})}$.
Their technique to achieve this result is a combination of elements of Baker's approach, an extension of Menger's theorem and an intricate divide and conquer scheme.

A combination of the techniques from~\cite{DBLP:conf/icalp/BodlaenderNZ16} and~\cite{pattern} gave a $2^{O(k/\log k)}$ time randomized algorithm for detecting occurrences of a given connected pattern.

\subsection{Our Approach}
We now briefly describe the high level intuition behind our approach to obtain our main result, Theorem~\ref{thm:main}.
As mentioned above, our approach employs aspects of the relatively new works of Fomin et al.~\cite{pattern} and Bodlaender et al.~\cite{DBLP:conf/icalp/BodlaenderNZ16}, but also uses more classic techniques such as Baker's partitioning to reduce the treewidth, as already proposed by Eppstein for subgraph isomorphism~\cite{DBLP:journals/jgaa/Eppstein99}. We now give a brief outline of our approach, with an emphasis on our main innovations.

\paragraph{Detecting Patterns: Sparsifying Balanced Cycle Separators}
We follow the approach from~\cite{pattern} that employs a Menger-like lemma (Lemma~\ref{lem:duality}) as a crucial ingredient, but we employ this lemma differently.
In Algorithms~\ref{alg:mainreduction1} and~\ref{alg:halvewaste} we will use the a more involved version to prove Lemma's~\ref{lem:mainred} and~\ref{lem:cleanstep}.

We start by preprocessing the graph via a standard argument~\cite{DBLP:journals/jacm/Baker94, DBLP:journals/jgaa/Eppstein99} to ensure it is $O(k)$-outerplanar. 
This implies we can find a small balanced cycle separator $C$ (after triangulation).

Then we use that for given any balanced (with respect to an unknown weight function) cycle separator $C$, we can construct a family of $\quasipoly(k)$ balanced cycle separators such that at least one cycle of the output family has small intersection with the (unknown) pattern $P$.

To obtain this family, we partition the cycle in four equally-sized consecutive parts $C^\pointleft,C^\pointdown,C^\pointright,C^\pointup$, inspired by a proof of the planar grid-minor theorem (following a version of the proof by Grigoriev~\cite{DBLP:journals/dmtcs/Grigoriev11}).
Then we consider $G'$ which either is the interior or the exterior of $C$ (depending on which has higher weight, and we can try both if the weight is unknown).
Applying Lemma~\ref{lem:duality} in $G'$, we either get a family of mutually disjoint $(C^\pointleft-C^\pointright)$-separators (which are $(C^\pointdown-C^\pointup)$-paths) or nearly-disjoint $(C^\pointdown-C^\pointup)$-separators (which are $(C^\pointleft-C^\pointright)$-paths).
In either case, we (non-deterministically) guess a path $S_i$ with little intersection with the pattern $P$ (which exists as the paths have limited mutual overlap).
Now we form two different cycles from $S_i \cup C$. The cycle with smallest weight in its exterior can be shown to be sufficiently balanced.
Repeating the procedure $O(\log(k))$ times suffices to prove the lemma.

\paragraph{Detecting Patterns: Acquiring Balance}
If we would apply the above approach recursively in a direct way to obtain a good tree decomposition-like divide and conquer scheme for running a dynamic programming to detect patterns, we quickly would arrive at running times of the type $n^{\log(k)}2^{\tilde{O}(\sqrt{k})}$, for problems like directed longest path or longest cycle on $k$ vertices.
This is already a very strong indication that a $2^{\tilde{O}(\sqrt{k})}$ running time is within reach, and indeed the following simple additional new idea allows such running time:
When given a balanced cycle separator $C$, we first guess whether $C$ has at most or at least $\sqrt{k}$ vertices from the pattern $P$.

If $C$ has at most $\sqrt{k}$ vertices from $P$, there are only $\tbinom{O(k)}{\sqrt{k}}$ possibilities for the image of the mapping of $C$ to $P$, and the associated dynamic programming table will be small enough. Thus $C$ can be used to decompose the problem into two subproblems with both only a constant fraction of the vertices of $G$. Otherwise, the assertion that $C$ has at least $\sqrt{k}$ pattern vertices can be used to construct another cycle $C'$ that has $\Omega(\sqrt{k})$ vertices from $P$ in both its interior and exterior. Then we use $C'$ as basis for the procedure outlined above to construct a family of cycles that separate at least $\Omega(\sqrt{k})$ pattern vertices. Oversimplifying things, the number of recursive calls $T(n,k)$ of a divide and conquer scheme applying this strategy exhaustively in terms of graphs with $n$ vertices and patterns with $k$ vertices satisfies the following upper bound:
\[
	T(n,k) \leq 2T(2n/3,k)+\quasipoly(k)\cdot T\left(n,k-\Omega(\sqrt{k})\right) = n^{O(1)}\left(\log(n)\right)^{\tO(\sqrt{k})} \leq 2^{\tilde{O}(\sqrt{k})}n^{O(1)},
\]
whether the latter upper bound can be shown by a case distinction on whether $\log(n) \leq k$.

\paragraph{Counting Patterns: Efficient Inclusion-Exclusion}
Note it is even not clear how to make the preprocessing step by Eppstein~\cite{DBLP:journals/jgaa/Eppstein99} and Baker~\cite{DBLP:journals/jacm/Baker94} to make the graph $O(k)$-outerplanar work in the counting setting as the natural extension (sum over all blocks of the partition, remove the block and count the number of pattern occurrences) will over count pattern occurrences.
Moreover, extending the dynamic programming table by keeping track of which blocks of the partition vertices have been selected as done in~\cite{DBLP:conf/stacs/Dorn10,DBLP:journals/jgaa/Eppstein99} increases the number of table entries to $2^{\Omega(k)}$.

Instead, we present a new approach based on inclusion-exclusion. Indeed, to avoid over count it is natural to compensate by summing over all subsets of the $k$ blocks in the partition and count the number of pattern occurrences exactly using inclusion-exclusion.
To avoid summing over all $2^{O(k)}$ sums in the inclusion-exclusion formula, we make the crucial observation that it's summands are algebraically dependent in a strong sense.
We show that we only need to compute pattern occurrences in $O(k^2)$ subgraphs that are $O(k)$-outer planar, and can evaluate the inclusion-exclusion formula in polynomial time given these values.
We call this (to our best knowledge, new\footnote{Let us remark that inclusion exclusion was used before for counting problems in planar graphs by Curticapean in~\cite{DBLP:conf/esa/Curticapean16} to reducing counting non-perfect matchings to non-perfect matchings, but in a very different way.}) idea \emph{Efficient inclusion-exclusion}.

We point out that without this new idea it would even be hard to get very special subcases of our general theorem, such as to count $k$-vertex independent sets in subgraphs of grids in $2^{\tilde{O}(\sqrt{k})}$ time.

\paragraph{Counting Patterns: Combinining all ideas}
To prove Theorem~\ref{thm:main} in its full generality, we combine all above new insights with the isomorphism check as exploited by Bodlaender et al.~\cite{DBLP:conf/icalp/BodlaenderNZ16}.
But to combine all above steps, still a number of technical hurdles need to be overcome.

First, when the step in which we (non-deterministically) guess a path $S_i$ with little intersection is replaced with summing over all $i$, we will over count.
We resolve this in different ways depending on whether the set of paths is completely disjoint or nearly-disjoint.
In the first case we can avoid over counting by keeping track of that we need to intersect some paths in the recursion.
We implement idea by associating a set of \emph{monitors} with a recursive call.
Specifically, a monitor is a set $M$ with two associated integers $M_{\low}$ and $M_{\upp}$.
We distinguish `small' monitors (with $\poly(k)$ vertices), and large monitors (with an unbounded number of vertices).
Given a set of monitors $\cM$ in a subproblem, we count, for every vector $r \in \mathbb{Z}^\cM$ such that $M_{\low}\leq r_M \leq M_{\upp}$, the number of occurrences of $P$ in $G$ on vertex set $X \subseteq V(G)$ such that $|X \cap M|=r_M$ for every $\cM \in \cM$.
In the second case, we apply the efficient inclusion-exclusion idea (in a slightly more technical, but essentially same, way as we did to reduce the outerplanarity of $G$)

Before we continue with sparsifying a balanced cycle separator, we need to ensure that the subproblem is `clean'.
Specifically, we need that the number of monitors and the number of `boundary vertices' (i.e. vertices of which we need to track how the pattern maps to them) are $\tO(\sqrt{k})$. 
To ensure this, we employ a \emph{cleaning step} that aims at sparsifying separators that balance the number of vertices in `small' monitors and boundary vertices.
After $O(\log k)$ of such steps, there will only $\tO(\sqrt{k})$ of such vertices left.

\subsection*{Organization}
Notation, useful standard tools, and other preliminaries are described in Section~\ref{sec:prel}.
In Section~\ref{sec:cis} we set up main building blocks of Theorem~\ref{thm:main}, which we subsequently prove in Section~\ref{sec:mainred}.

\section{Preliminaries}
\label{sec:prel}
\paragraph{Notation and Basic Definitions}
With a triangulated graph we mean a graph with a given embedding in which all faces are of size $3$.
Let $\mathbb{N}:=\{1,2,\ldots\}$, $[k]:=\{1,\ldots,k\}$, denote $a \% b$ for the remainder of $a/b$, and $ \equiv_p$ for being congruent mod $p$.
We use $\binom{X}{\leq s}$ and $\binom{X}{s}$ for all subsets of $X$ of size at most and respectively equal to $s$.
In this paper $O^*(\cdot)$ suppresses factors polynomial in the problem instance size.
Let $G$ be an undirected graph.
Whenever $X\subseteq V(G)$, we let $G-X$ denote $G[V(G) \setminus X]$.
Similarly, if $X\subseteq E(G)$, we let $G-X$ denote the graph $(V(G),E(G) \setminus X)$.
If $w:V(G)\rightarrow \mathbb{R}$, we shorthand $w(G):= w(V(G)):=\sum_{v \in V(G)}w(v)$.
If $\cF \subseteq 2^U$ is a set family and $X \subseteq U$ we denote $\cF[X]:= \{ F \cap X: F \in \cF \}$.
If $A,B$ are sets we denote $A^B$ for all vectors indexed by $B$ with values from $A$.
We use both the $a_i$ and $a[i]$ index notation in this paper, to occasionally avoid subscripts.
Given two vectors $a,b \in \mathbb{Z}^B$, we let denote $a \preceq b$ that $a_x \leq b_x$ for every $x \in B$.
If $v \in V(G)$, $\dist(v_0,v)$ denotes the length of (that is, the number of edges on) the shortest path from $v_0$ to $v$.

If $X \subseteq V(G)$, we denote $\partial_G X:= \{x \in X: N(x) \not\subseteq X\}$. If $G$ is clear from the context it will be omitted.
We use $\tO(f(k))$ to omit $\polylog(k)$ factors, let $\poly(k)$ denote all functions of the type $k^{O(k)}$ and $\quasipoly(k)$ denote all functions of the type $k^{\polylog(k)}$.

\paragraph{Functions}
Given a function $f: A \rightarrow B$ and $b \in B$, we let $f^{-1}(b)=\{a \in A:f(a)=b \}$.
We let $f: A \into B$ denote that $f$ is injective (that is $f(a)=f(a')$ implies that $a=a'$).
If $f': A' \rightarrow B$ for a superset $A \subseteq A'$, we say $f'$ extends $f$ if $f'(a)=f(a)$ for every $a \in A$. In this case we also say $f$ is the projection of $f'$ on $A$.
If $f: A\rightarrow B$ and $g: C \rightarrow B$, we say $f$ and $g$ \emph{agree} if $f(x)=g(x)$ for every $x \in A \cap B$.
If $f^{-1}(b)$ is a singleton set, we may also interpret it as a single element of $A$.
If $G$ and $P$ are undirected graphs, we denote $\sub(P,G)$ for the set of injective functions $\{f: V(P) \into V(G): \{v,w\} \in E(P) \rightarrow \{f(v),f(w)\} \in E(G) \}$, and $\ind(P,G)$ for the set of injective functions $\{f: V(P) \into V(G): \{v,w\} \in E(P) \leftrightarrow \{f(v),f(w)\} \in E(G) \}$.
A bijection $f: V(G) \rightarrow V(G)$ is an \emph{isomorphism} if $\{a,b\} \in E(P)$ if and only if $\{f(a),f(b)\} \in E(G)$.
An \emph{automorphism} is an isomorphism from a graph $G$ to itself.
We let $\auto(G)$ denote the set of automorphisms of $G$.

\paragraph{Separations and Their Isomorphism Classes}
A \emph{colored graph} if a pair $(G,c)$ where $G$ is a graph and $c:V(G)\rightarrow \mathbb{N}$ is a coloring function.
Two colored graphs $(G,c)$ and $(G',c')$ are \emph{isomorphic} if there exists an isomorphism $f$ from $V(G)$ to $V(G')$ such that $c(v)=c'(f(v))$ for every $v \in V$. 
It is known that testing whether two colored $n$-vertex planar graphs are isomorphic can be done in $\quasipoly(n)$ time: Using standard techniques (see e.g.~\cite[Theorem 1]{thesisSchweitzer}) one can reduce planar colored subgraph isomorphisms to normal planar subgraph isomorphism, which can be solved in planar graphs in polynomial time~\cite{Hopcroft:1974:LTA:800119.803896}. 
By the same standard reduction from colored subgraph isomorphism to subgraph isomorphism, and the canonization algorithm for planar subgraph isomorphism by Datta et al.~\cite{DBLP:conf/coco/DattaLNTW09}, we also have the following:

\begin{theorem}
	There exists a polynomial time algorithm $\can(G,c)$ that given a colored planar graphs $(G,c)$ outputs a string $s$ such that $\can(G,c)=\can(G',c')$ if and only if $(G,c)$ is isomorphic to $(G',c')$.
\end{theorem}

A \emph{separation} of a graph $G$ is pair of two subsets $X,Y \subseteq V(G)$ with no edges between $X \setminus Y$ and $Y \setminus X$ in $G$.
We say $X \cap Y$ is the \emph{separator} of this separation and that $|X \cap Y|$ is the \emph{order} of this separation.
If $Z\subseteq V(G)$, we say $(X,Y)$ is \emph{below} $Z$ if $X \subseteq Z$.
Two separations $(X,Y)$ and $(X',Y')$ are \emph{isomorphic} if $X\cap Y=X'\cap Y'$ and there exists an automorphism $a \in \auto(G)$ such that $a(v)=v$ for every $v \in X \cap Y$.
We let $\sepsa{l,P,Z}$ (respectively, $\sepse{l,P,Z}$) denote an (arbitrarily fixed) maximal set of pair-wise non-isomorphic separations of $P$ below $Z$ of order at most $l$ (respectively, exactly $l$).
We also shorthand $\sepsa{l,Z}=\sepsa{l,P,Z}$, $\sepse{l,Z}=\sepse{l,P,Z}$ since the input pattern will be fixed throughout this paper, and shorthand $\sepsa{l}=\sepsa{l,P,V(P)}$ and $\sepse{l}=\sepse{l,P,V(P)}$.
Define $\mult((X,Y),Z)$ to be the number of separations $(X',Y')$ of $P$ below $Z$ such that $\can(X',Y')=\can(X,Y)$.

\begin{lemma}
	Given $P, l$ and $Z$, we can enumerate $\sepsa{l,Z}$ and $\sepse{l,Z}$ in $|\sepsa{l,Z}||V(P)|^{O(1)}$ time.
	In the same time we can also compute $\mult((X,Y),Z)$ of each separation $(X,Y) \in \sepsa{l,Z}$.
\end{lemma}
\begin{proof}
	Iterate over all $\binom{k}{\leq l}$ possibilities for the separator $S= X \cap Y$.
	Subsequently, for each connected component $V_i$ of $P - S$, create a colored graph $(G_i,c_i)$ on vertex set $V_i \cup S$ with colors $1,\ldots,|S|$ assigned to the vertices in $S$ and a single color to all vertices in $V_i$.
	To enumerate $\sepse{l,Z}$ and compute $\mult((X,Y),Z)$, we can iterate over all possibilities of $S$, label each connected component with their canonical string $\can(G_i,c_i)$. And compute the number $q_s$
	which we define as the number of connected components satisfying $\can(G_i,c_i)=s$, and $q^Z_s$ be the number of connected components with $\can(G_i,c_i)=s$ and $V_i \subseteq Z$.
	
	Subsequently, we enumerate over all non-negative vectors $q^X,q^Y$ such that $q^X_s+q^Y_s=q_s$ and $q^X \preceq q^Z$.
	Note that each such $q^X_s$ uniquely defines gives a non-isomorphic separation $(X,Y)$ in which $X \cap Y = S$ and $P - S$ contains exactly $q^X_s$ connected components $V_i$ such that $\can(V_i \cap S,c_i)=s$. 
	For each such $S$ and $q^X_s$, we add the separation $(X,Y)$ to $\sepse{l,Z}$.
	Moreover, by the above discussion it also follows that $\mult((X,Y),Z)=\prod_{i}\binom{q^Z_s}{q^X_s}$, as for every connected component with canonical string $s$ we have $\binom{q^Z_s}{q^X_s}$ options to choose the $q^X_s$ connected components included in $X$ from the $q^Z_s$ available connected components.
\end{proof}

\begin{lemma}\label{lem:boundsigmacon}
	If $P$ is connected and has bounded degree, $|\sepsa{\sqrt{k}}| = 2^{O(\sqrt{k}\log k)}$.
\end{lemma}
\begin{proof}
	There are at most $\binom{k}{\sqrt{k}}$ possibilities for the set $S = X \cap Y$. The graph $P-S$ has at most $O(\sqrt{k})$ connected components, and these can be distributed among $X$ and $Y$ in $2^{O(\sqrt{k})}$ ways.
\end{proof}

The following lemma is a direct consequence of the proof from Section~3.3 of~\cite{DBLP:conf/icalp/BodlaenderNZ16}:
\begin{lemma}\label{lem:boundsigmaBod}
	For any planar $P$,  $|\sepsa{\sqrt{k}}| = 2^{O(k / \log k)}$.
\end{lemma}

\paragraph{Planar Graphs}
A \emph{cycle} of a graph $V(G)$ is a sequence $v_1,\ldots,v_l \in V(G)$ such that $\{v_l,v_1\} \in E(G)$ and $\{v_i,v_{i+1}\} \in E(G)$ for every $i=1,\ldots,l-1$. The cycle $C$ is \emph{simple} if every vertex appears at most once in it. The \emph{length} of $C$ is $\ell$.
If $G$ is planar and its embedding in $\mathbb{R}^2$ is clear from the context, we let $\int_G(C)$ denote the subgraph of $G$ consisting of $C$ and all edges and vertices enclosed by $C$ (the `interior of $C$'), and let $\ext_G(C)$ denote the subgraph of $G$ consisting of $C$ and all vertices and edges not enclosed by $C$ (the `exterior' of $C$).
The strict interior (exterior) of $C$ is the interior (exterior) except $C$, and are denoted with $\sint_G(C)$ and $\sext_G(C)$.

\begin{definition}
	For a graph $H$ and a vertex $u$ in $H$, by $\reach(u,H)$ we denote the set of vertices of $H$ reachable from $u$ in $H$. If $u$ is not in $H$, $\reach(u,H)$ should be read as the empty set.
	If $U\subseteq V(H)$ this is extended in the natural way, i.e. $\reach(U,H):=\cup_{u \in U}\reach(u,H)$. Suppose $G$ is a connected graph, and $s,t$ are different vertices of $G$. An \emph{$(s,t)$-separator} is a subset $S$ of vertices of $G$ such that $s,t \notin S$ and $t \notin \reach(s,G\setminus S)$. Moreover, $S$ is said to be \emph{a minimal $(s,t)$-separator} if no strict subset of $S$ is an $(s,t)$-separator, and $S$ is \emph{minimal} if it is a minimal $(s,t)$-separator for some $s,t$.
\end{definition}

\paragraph{Outerplanarity} A planar embedding of a graph is $1$-outerplanar if all its vertices are on the outerface, and it is $k$-outerplanar if after the removal of the vertices on the outerface an $(k-1)$-outerplanar embedding remains. A graph is $k$-outerplanar is it admits a $k$-outerplanar embedding.
We will need the following facts on outerplanarity:

\begin{lemma}[\cite{BIEDL2015275}]\label{lem:outerplanar}
	Every $k$-outerplanar graph can be triangulated to a $(k+1)$-outerplanar graph.
\end{lemma}


\begin{lemma}[\cite{DBLP:journals/jacm/Baker94}]\label{lem:outerplanarityreduction}
	There is an algorithm that given a planar graph $G$ and integer $k$, outputs subsets $A_1,\ldots,A_{k+1} \subseteq V(G)$ such that $G[A_i]$ is $k$-outerplanar for every $i$ and for every $P \subseteq \binom{V(G)}{k}$ there exists an $i$ such that $P \subseteq A_i$.
\end{lemma}

\begin{lemma}[\cite{DBLP:journals/tcs/Bodlaender98}]\label{lem:twkout}
	Given a $k$-outerplanar graph $G$, we can construct a treedecomposition of $G$ of width $O(k)$ in polynomial time.
\end{lemma}

\paragraph{Inclusion Exclusion}
If $\cA= \{A_1,\ldots,A_m \} \subseteq U$ is set family over a universe $U$, then 
\begin{equation}\label{eq:ie}
|\bigcup_{i=1}^{m} A_i|= \sum_{\emptyset \neq  R \subseteq [m]} (-1)^{|R|-1} |\bigcap_{i \in R}A_i|.
\end{equation}
Note that $R$ resembles a set of `required sets'.

\paragraph{Balanced Separators}

A $\beta$-proper weight assignment $w$ is an assignment of weight to vertices summing to $1$ with all weights being at most $\beta$.
An $\beta$-balanced cycle separator for $w$ in $G$ is a cycle $C$ such that the weight of all vertices in the strict interior of $C$ is at most $\beta$ and the weight of all vertices in the strict exterior of $C$ is at most $\beta$.
If $X$ is a set of vertices, we say $S$ is balanced for $X$ if it is balanced for the weight function that assigns $1/|X|$ to all vertices of $X$ and weight $0$ to all other vertices.
We use the following lemma:

\begin{lemma}[Folklore (see e.g~Lemma 5.3.2. in~\cite{planarity})]\label{lem:balcyc}
	There is a linear-time algorithm that, given a triangulated graph, spanning tree $T$ of $G$, and a $\tfrac{1}{4}$-proper assignment to vertices, returns a nontree edge $\hat{e}$ such that the fundamental cycle of $\hat{e}$ with respect to $T$ is a $\tfrac{3}{4}$-balanced cycle separator for $w$ in $G$.
\end{lemma}
In particular, the lemma implies we can find such separators of size $d$ in polynomial time whenever the diameter of $G$ is at most $d$ or when the graph is $d/2$-outerplanar and triangulated.

\paragraph{A Menger-like Theorem for Nearly Disjoint Paths}

\begin{definition}
	A sequence $\sigma=(S_1,\ldots,S_\ell)$ of $(s,t)$-separators is called an $(s,t)$-separator chain if for each $1 \leq i <j \leq \ell$, the following holds:
	\[
	S_i\setminus S_{j} \subseteq \reach(s,G - S_{j} ) \qquad \text{and} \qquad S_{j}\setminus S_i \subseteq \reach(t,G - S_i).
	\]
	If $\sigma$ is clear from the context, we denote $\priv(S_i):= S_i \setminus \bigcup_{i' \neq i} S_{i'}$ for the private vertices of $S_i$, and $\pub(S_i) := S_i \cap \bigcup_{i' \neq i} S_{i'}$ for the public vertices of $S_i$. If $\pub(S_i)=\emptyset$ for all $i$ we call $\sigma$ \emph{disjoint}.
\end{definition}
Somewhat counter-intuitively, for $\sigma$ to be a separator chain some connectivity might be required that is not present in the graph, but assuming these connections exist will be notationally convenient as it gives some sense of linear order in the chain. We frequently will be interested in separator chains in graph not having the required connectivity and fix this issue by working with appropriate super graphs, which is allowed as this only filters out some separators.

A crucial tool in our approach the following useful lemma from~\cite{pattern}, already designed specifically to find patterns in planar graphs in sub-exponential parameterized time.
\begin{lemma}[\cite{pattern}]\label{lem:duality}
	There is a polynomial time algorithm that, given a connected graph $G$, a pair $s,t \in V(G)$ of distinct vertices, and integers $p,q \in \mathbb{N}$, outputs one of the following structures in $G$:
	\begin{enumerate}[(a)]
		\item A chain $(S_1,\ldots,S_p)$ of $(s,t)$-separators with $|S_j| \leq 2q$ for each $j \in [p]$,
		\item A sequence $(P_1,\ldots,P_q)$ of $(s,t)$-paths with $|\pub(S_i) \setminus \{s,t\}| \leq 4p$ for each $i \in [q]$.
	\end{enumerate}
\end{lemma}

\paragraph{Crossing Paths}
We use some standard definitions and tools for crossing paths in planar graphs:

\newcommand{\sort}{\ensuremath{\mathtt{sort}}}

\begin{definition}[\cite{10.1007/11830924_10}]\label{def:cross}A path $P$ \emph{crosses} another path $P_0$ if there exists a bounded connected region $X$ in $\mathbb{R}^2$ with the following properties: $P$ and $P_0$ each cross the boundary of $X$ exactly twice and these crossings are interleaved. A set of paths is said to be \emph{non-crossing} if every pair of paths is distinct and non-crossing.
\end{definition}

It is easy to modify a set of paths $\mathcal{P}$ with common endpoints, into another set of paths $\mathcal{P}'$ in which every edge occurs equally often as in $\mathcal{P}$ such that $\mathcal{P}'$ is non-crossing in polynomial time. See also~\cite{10.1007/11830924_10} for more details. A set of non-crossing paths with common endpoints in an embedded graph can be sorted in a natural way: sort all edges in clockwise order, and order the paths lexicographically according to this order. We denote the algorithm that does this for us $\sort(\mathcal{P})$.

\begin{definition}[Alignment of cycle]
	Let $C$ be a cycle and $V(C)=\{v_1,\ldots,v_l\}$ be an arbitrary but fixed consecutive ordering of its vertices. Suppose $h<i<j<k<l$, and let 
	\[
		C^{\pointleft}=\{v_1,\ldots,v_i\}, \quad C^{\pointdown}=\{v_{i+1},\ldots,v_{j}\}, \quad C^{\pointright}=\{v_{j+1},\ldots,v_{k}\}, \quad C^{\pointup}=\{v_{k+1},\ldots,v_l\}.
	\]
	Then $C^{\pointleft}, C^{\pointdown}, C^{\pointright}, C^{\pointup}$ form an \emph{alignment} of $C$.
\end{definition}


We will use the sorting step to obtain the following consequence of Lemma~\ref{lem:duality}.
\begin{lemma}\label{lem:chaineitherway}
	Let $G$ be an inner-triangular graph with outer boundary $C$, and let $C^{\pointleft},C^{\pointdown},C^{\pointright},C^{\pointup}$ be an alignment of $C$. 
	Let $\hat{G}$ be obtained by adding vertices $v^{d}$ adjacent to all vertices $C^d$ for every direction $d \in \{\pointleft,\pointdown,\pointright,\pointup\}$.
	There is a polynomial time algorithm $\mathtt{menger+}$ that, given $G,C, C^{\pointleft},C^{\pointdown},C^{\pointright},C^{\pointup}$ and integers $p,q$ either finds a chain of
	\begin{enumerate}
		\item disjoint $(v^\pointdown,v^\pointup)$-separators $(S_1,\ldots,S_p)$ in $\hat{G}$ with $|S_j| \leq 2q$ for each $j \in [p]$, or 
		\item $(v^\pointleft,v^\pointright)$-separators $(S_1,\ldots,S_q)$ in $\hat{G}$ with $|\pub(S_i) \setminus \{v^\pointleft,v^\pointright\}| \leq 4p$ for each $i \in [q]$.
	\end{enumerate}
\end{lemma}
\begin{proof}
	Apply Lemma~\ref{lem:duality} with $v^\pointdown,v^\pointup$ as given.
	If a chain of $(v^\pointdown,v^\pointup)$ separators is found we are done immediately.
	Otherwise, $\sort(S_1,\ldots,S_p)$ gives a non-crossing ordered set of paths from $v^\pointdown,v^\pointup$.
	As each such a path is a $(v^\pointleft,v^\pointright)$-separator this is also a chain of $(v^\pointleft,v^\pointright)$-separators by the ordering.
\end{proof}

\paragraph{Tree decompositions and treewidth}
A \emph{tree decomposition} of a graph~$G$ is a pair~$\treedecomp=(T,\{B_x\}_{x \in V(T)})$ in which $T$ is a tree, and $B_x \subseteq V_x$ are subsets 
such that $\bigcup_{x \in V(T)} B_x = V$ with the following properties:
(i) for any $uv \in E$, there exists an~$x \in V(T)$ such that 
$u,v \in B_x$,
(ii) if $v \in B_x$ and $v \in B_y$, then $v \in B_z$ for all $z$ on 
the (unique) path from $x$ to $y$ in $T$.
%
%
%

The \emph{width}~$tw(\treedecomp)$ of a tree decomposition is the maximum bag size minus one, and the treewidth of a graph $G$ is the minimum treewidth over all nice tree decompositions of $G$. 

\section{Monitors, Subproblems and Helper Reductions}\label{sec:cis}

In this section we set up machinery that will be used in Section~\ref{sec:mainred}. In Subsection~\ref{subsec:monsubred}, we introduce definitions that will facilitate the presentation of the algorithm in subroutines.
In Subsection~\ref{subsec:preprocreasy}, we show how to ensure the input graph is $k$-outerplanar and how to solve subproblems with $\tO(\sqrt{k})$ pattern vertices in $2^{\tO(\sqrt{k})}$ time.
The latter will form a base case for our divide and conquer scheme leading to the proof Theorem~\ref{thm:main}.
In Subsection~\ref{subsec:sepsystemsreduction} we provide more technical lemma's that use systems of separators for reductions, and how to acquire a balanced separator.

\subsection{Monitors, Subproblems and Reductions}\label{subsec:monsubred}
A monitor is an object that `monitors' a particular set of vertices $M$, in the sense that for some range of intersection sizes $M_\low \leq  i \leq M_\upp $ one counts all pattern occurrence with exactly $i$ vertices in $M$:

\begin{definition}[Monitor]
	A \emph{monitor over $U$} is a triple $(M,M_\low,M_\upp) \in 2^U \times \mathbb{N}_{\geq 0} \times \mathbb{N}_{\geq 0}$.
	If $\cM$ is a set of monitors over $U$ and $M \subseteq U$, we denote
	\begin{align*}
		\low_{\cM}(M) &:=\max \{ M_\low : (M,M_\low,M_\upp) \in \cM \},\\
		\upp_{\cM}(M) &:=\min \{ M_\upp : (M,M_\low,M_\upp) \in \cM \}.
	\end{align*}
	 We let $\feas(\cM)$ denote the set of vectors $r \in [k]^{\cM}$ with $r_{(M,M_\low,M_\upp)} \in [M_\low,M_\upp]$ for every monitor $(M,M_\low,M_\upp) \in \cM$.
	 We say a monitor is \emph{small} if $|M_\upp| > 0$ and $|M| \leq k^4$, and it is \emph{large} if $|M_\upp| > 0$ and $|M| > k$.
	 We denote $\smallM(\cM)$ for the set of small monitors in $\cM$,  and $\largeM(\cM)$ for the set of laronitorge monitors.
	 If $f: A \rightarrow U$ and $r \in \feas(\cM)$, we say $f \picks r$ if for every $M  \in \cM$ it holds that $|f(A)\cap M|=r_M$.
\end{definition}

	
Now we define a `subproblem', which corresponds to a recursive call of the divide and conquer scheme we will employ to prove Theorem~\ref{thm:main} in Section~\ref{sec:mainred}.

\begin{definition}[Subproblem]
A \emph{subproblem} is a tuple $\pi=(G,B,\cM)$ where $G$ is a graph, $B \subseteq V(G)$, and $\cM$ is a family of monitors over $V(G)$.
The \emph{answer} to $\pi$ is the vector $a$ indexed by every $r \in \feas(\cM)$, $(X,Y) \in \sepsa{\upp_{B}(\cM),P}$ and function $f: X \cap Y \rightarrow B$ such that
	\[
		\ans{a}{r,(X,Y),f} := |\{ g \in \ind(P[X],G): g \textrm{ extends } f \wedge g \picks r\}|.
	\]
\end{definition}

We index answers also by non-canonical separations; the value in the answer then can be deduced algorithmically fast via finding the value in the answer with corresponding separator in the same equivalence class via basic data-structures.

Note that there are $k^{|\cM|}$ possible values for $r$, $|\sepsa{\upp_{\cM}(B),P}|$ options for $(X,Y)$, and $|B|^{\upp_{B}(\cM)}$ options for $f$.
While generating subproblems, we will therefore ensure that $|\cM| = \tO(\sqrt{k}) + O(\log n/\log k)$, $|B| \leq k^{O(1)}$ and that $\upp_{B}(\cM) \leq \tO(\sqrt{k})$  for any invoked subproblem.

The following lemma allows us to only compute $a[r,(X,Y),f]$ for a maximal set of non-isomorphic separations $(X,Y)$.
\begin{lemma}
	If $(X,Y)$ is isomorphic to $(X',Y')$, then $a[r,(X,Y),f]=a[r,(X',Y'),f]$.
\end{lemma}
\begin{proof}
	Suppose there is an $\alpha \in \auto(P)$ such that $\alpha(v)=v$ for every $v \in X \cap Y$, and that $g \in \ind(P[X],G)$ such that $g \textrm{ extends } f$, and $\forall M \in \cM: |g^{-1}(M)|= r(M)$ (i.e. $g$ is counted in $a[r,(X,Y),f]$).
	Then $g'=f\circ\alpha$ contributes to $a[r,(X',Y'),f]$ and since $g'$ and $\alpha$ determine $g$ it follows that $a[r,(X \cap Y),f]=a[r,(X' \cap Y'),f]$.
\end{proof}

A recursive step in our divide and conquer scheme that reduces subproblems to supposedly easier subproblems is formalized as follows:

\begin{definition}[Reduction]
	A \emph{reduction} from a subproblem $\pi=(G,B,\cM)$ to a set of subproblems $\pi_1,\ldots,\pi_l$ is an algorithm that, (1) given $\pi$, outputs $\pi_1,\ldots,\pi_l$, and (2) given the answers of the subproblem $\pi_1,\ldots,\pi_l$ outputs the answer to $\pi$.  
	Denoting $\pi_i = (G_i,B_i,\cM_i)$, the reduction is said to be \emph{strict} if  $\upp_{\cM_i}(V(G_i)) \leq \upp_{\cM}(V(G))$, $|V(G_i)\setminus B_i| \leq |V(G)\setminus B|$ for every $i$, and both algorithms run in time polynomial in the input and output.
\end{definition}

We refer to $V(G) \setminus B$ as the \emph{non-boundary vertices} of the subproblem and to $\upp_{\cM}(V(G))$ as the \emph{mapped pattern vertices}. The number of those vertices will be a primary c

\subsection{Preprocessing and Solving Subproblems with few Pattern Vertices}\label{subsec:preprocreasy}

We now show how to preprocess the input to ensure $G$ is $k$-outerplanar.
As mentioned in the introduction, this is where one of new ideas dubbed `efficient inclusion-exclusion' is required.
It makes use of the following lemma that shows how to quickly evaluate a quantity that will later arise from the use of the inclusion-exclusion formula.
We state it separately here so we can reuse it in a later second application of efficient inclusion-exclusion.


\begin{lemma}\label{lem:techDP}
	Given a set $A$, an integer $h$ and a value $T[x,x'] \in \mathbb{Z}$ for every $x,x' \in A$, the value
	\begin{equation}\label{eq:sumproduct}
		\sum_{x_1,\ldots,x_h \in A} \prod_{i=1}^{h-1} T[x_i,x_{i+1}]	
	\end{equation}
	can be computed in $O(h|A|^2)$ time.
\end{lemma}
\begin{proof}
	For $x \in A$, define $T_h[x] := \left( \sum_{x_1,\ldots,x_{h-1} \in A} \prod_{i=1}^{h-2} T[x_i,x_{i+1}] \right) T[x_{h-1},x]$.	For $h=2$, we see that $T_h[x]=\sum_{x_1}T[x_1,x]$ can be computed in $|A|$ time.
	For $h>3$, note that
	\begin{align*}
	T_h[x] 	&= \left(\sum_{x_1,\ldots,x_{h-1} \in A} \prod_{i=1}^{h-2} T[x_i,x_{i+1}] \right) T[x_{h-1},x] \\
	&= \left(\sum_{x_1,\ldots,x_{h-2} \in A} \prod_{i=1}^{h-2} T[x_i,x_{i+1}] \right)\sum_{x_{h-1} \in A}T[x_{h-2},x_{h-1}] T[x_{h-1},x] \\
	&= \sum_{x_{h-1} \in A}T_{h-1}[x_{h-1}] T[x_{h-1},x].
	\end{align*}
	Thus $T_{h}[x]$ can be computed in time $O(h|A|^2)$. As~\eqref{eq:sumproduct} equals $\sum_{x \in A}T_h[x]$, the lemma follows.
\end{proof}
By summing over all $x_1 \in A$ and applying the above Lemma, we also get the following corollary: 
\begin{corollary}\label{lem:techDPwrap}
	Given a set $A$, an integer $h$ and a value $T[x,x'] \in \mathbb{Z}$ for every $x,x' \in A$, the value
	\begin{equation}\label{eq:sumproductwrap}
	\sum_{x_1,\ldots,x_h \in A} T[x_h,x_1] \prod_{i=1}^{h-1} T[x_i,x_{i+1}]	
	\end{equation}
	can be computed in $O(h|A|^3)$ time.
\end{corollary}

Now we focus on reducing the outerplanarity
We are interested in the value $a_{\emptyset,\emptyset,(V(P),\emptyset)}$ in the answer of the subproblem $(G,\emptyset,\emptyset,0,k)$. We first show that to compute this answer, we may assume $G$ is $k$-outerplanar:
\begin{lemma}\label{lem:redkout}
	Given any planar graph $G$, one can in polynomial time compute graphs $G_{i,j}$ for $i,j \in [k+1]$ that are $k$-outerplanar graphs with the following property:
	given the answers to the subproblems $(G_{i,j},\emptyset,\{(V(G_{i,j}),0,k)\})$ for every $i,j \in [k+1]$, the quantity $|\ind(P,G)|$ can be computed in time $|\sepsa{0,V(P)}|n^{O(1)}$.
\end{lemma}
\begin{proof}
	Arbitrarily pick a vertex $v_0 \in V(G)$. Partition $V(G)$ into sets $V_1,\ldots,V_{k+1}$ where $V_i := \{v\in V(G): \dist(v_0,v)\equiv_{k+1} i\}$.
	It is easily seen that $G - V_i$ is $k$-outerplanar for every $i$.
	
	Let $f \in \ind(P,G)$. As $|f(P)|=k$ and the $V_i$'s form a partition, $f(P)$ will be disjoint from $V_i$ for some $i$ and we see that $|\ind(P,G)|$ equals
	\begin{equation}\label{eq:incexcappbak}
	\left|\bigcup_{i \in [k+1]} \ind\left(P,G - V_i\right)\right| =	\sum_{\emptyset \neq  R \subseteq [k+1]} (-1)^{|R|-1} \left|\ind\left(P,G-\bigcup_{i\in R}V_i\right)\right|,
	\end{equation}
	where the equality is by the inclusion-exclusion formula~\eqref{eq:ie}. For integers $i,j \in [k+1]$ we define
	\[
	G_{ij} = 
	\begin{cases}
	G\left[\bigcup_{k \in [i+1,\ldots,j-1]} V_k\right], & \text{if } i < j\\
	G- V_i, & \text{if } i=j\\
	G\left[\bigcup_{k \in [i+1,\ldots,k+1,1,j-1]} V_k\right], & \text{if } i > j.
	\end{cases}
	\]
	We see that if $R= \{z_1 < z_2 < \ldots < z_h\}$, then $\left|\ind\left(P,G-\bigcup_{i\in R}V_i\right)\right|$ equals
	\begin{equation}\label{eq:simprewrite}
	\sum_{\substack{X_0 \subseteq \ldots \subseteq X_h = V(P) \\ |\partial X_i|=0} } \big|\ind(P[X_0],G_{z_h,z_1})\big| \prod_{i=1}^{h-1} \big|\ind(P[X_i\setminus X_{i-1}],G_{z_i,z_{i+1}})\big|.
	\end{equation}
	To see this, note we can group the functions in $\ind(P,G-\bigcup_{i\in R}V_i)$ on their image $I$ which defines $Y_0,\ldots,Y_h$ as $Y_0 = I \cap V(G_{z_h,z_1})$ and $Y_i= I \cap G_{z_i,z_{i+1}}$ for $i>0$.
	This uniquely defines $X_j=\cup_{i \leq j}Y_j$.
	Thus we may count the functions by summing over all such $X_i$.
	Note that if $X_i \setminus X_{i-1}$ is the set of vertices mapped to $G_{z_i,z_{i+1}}$ then $\partial (X_i \setminus X_{i-1})=\emptyset$ as $G_{z_i,z_{i+1}}$ is a connected component.
	
	Now we can combine~\eqref{eq:incexcappbak} and~\eqref{eq:simprewrite} to get that $|\ind(P,G)|$ equals $\sum_{h=1}^{k+1} (-1)^{h-1} T_{h}$ where $T_h$ equals
	\begin{equation}\label{eq:incexcsimp}
		\sum_{z_1<\ldots<z_h}\sum_{\substack{X_0 \subseteq \ldots \subseteq X_h = V(P) \\ |\partial X_i|=0} } \big|\ind(P[X_0],G_{z_h,z_1})\big| \prod_{i=1}^{h-1} \big|\ind(P[X_i\setminus X_{i-1}],G_{z_i,z_{i+1}})\big|
	\end{equation}
	Here we can recognize~\eqref{eq:sumproductwrap}, with $A=\sepsa{0} \times [k+1]$ and 
	\[
		T[((X_1,Y_1),r_1),((X_2,Y_2),r_2)]= \mult{((X_1,Y_2),X_2)}\cdot a[0,X_2\setminus X_1,\emptyset].
	\]
	Therefore, Lemma~\ref{lem:techDP} implies~\eqref{eq:incexcsimp} can be evaluated in the claimed time.
\end{proof}

\paragraph{Sparse Separation Reduction}
We now present a simple, but for our approach fundamental, reduction.
Informally, it states a separator of $G$ that only contains few pattern vertices can be used for a reduction into two subproblems.
We state the underlying equation that relates the answers of subproblems separately to facilitate further use.

\begin{lemma}[Sparse Separation Reduction]\label{lem:sparsesepform}
	Suppose that $a$ is the answer of subproblem $\pi= (G,B,\cM)$, $(V_1,V_2)$ is a separation of $G$ such that $|\upp_{\cM}(V_1 \cap V_2)| \leq b$, $a_1$ is the answer of subproblem 
	$(G[V_1],(B \cup V_1) \cap V_2,\cM[V_1])$, and that $a_2$ is the answer of subproblem $(G[V_2],(B \cup V_2) \cap V_1,\cM[V_2])$.
	
	Then the following formula holds:
	\[	
		\ans{a}{r,(X,Y),f} = \sum_{
			\substack{
				(X_1,X_2) \in \sepsa{\upp_{\cM}(B)+ b,X} \\
				 f': \partial{X_1} \cup \partial{X_2}  \rightarrow  B \cup (V_1 \cap V_2)  \\
				 f' \extends f\\
				 r_1+r_2=r+r'
			  }
			  }
			 \mu((X_1,X_2),X)\cdot\ans{a_1}{r_1,X_1,f'_{|X_1 }} \cdot \ans{a_2}{r_2,X_2,f'_{|X_2}},
	\]
	where $r'_M := |f^{-1}(M \cap (V_1 \cap V_2))|$ for every $(M,M_\low,M_\upp) \in \cM$.
\end{lemma}
\begin{figure}
	\centering\resizebox{0.8\textwidth}{!}{
		\includegraphics{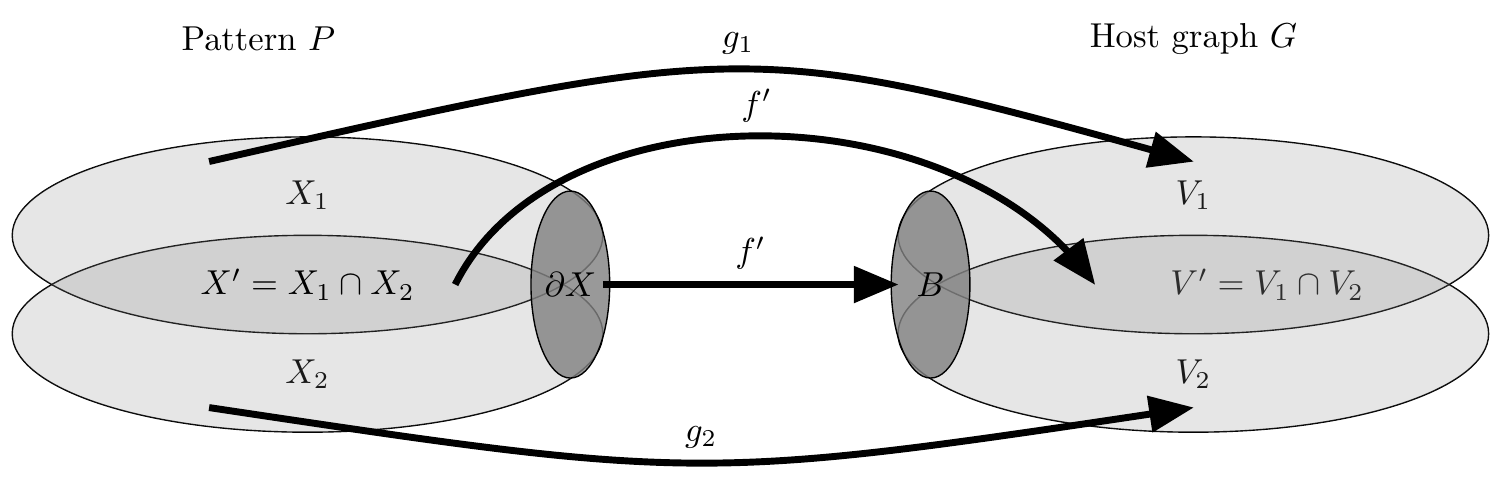}	
	}
	\caption{Mapping in the proof of Lemma~\ref{lem:sparsesepform}}
	\label{fig:mapping}
\end{figure}
\begin{proof}
	Since $\ans{a}{r,(X,Y),f}$ counts the number of mappings $g \in \ind(P[X],G)$, we can compute $\ans{a}{r,(X,Y),f}$ as the number combinations of appropriate mappings $g_1 \in \ind(P[X_1],G[V_1])$ and $g_2 \in \ind(P[X_2],G[V_2])$ (see also Figure~\ref{fig:mapping}).
	We sum over all $X_1 \subseteq X$ such that $|\boundary{X_1} \setminus \boundary{X}| \leq b$ which means that $X_1 \in \sepsa{\upp_{\cM}(B)+ b}$, and let $X_2= (X \setminus X_1) \cup \boundary{X_1}$. Then we have that
	\begin{align*}
		\ans{a}{r,(X,Y),f} &=\sum_{	\substack{	X_1 \subseteq X \\|\boundary{X_1} \setminus \boundary{X}| \leq b}}\sum_{r_1+r_2=r +r'} \sum_{\substack{f': (X_1 \cap X_2) \cup \partial X  \rightarrow  B \cup V'  \\ f' \extends f}} \ans{a_1}{r_1,X_1,f'_{|X_1 }} \cdot \ans{a_2}{r_2,X_2,f'_{|X_2}}\\
		&= \sum_{\substack{X_1 \in \sepsa{\upp_{\cM}(B)+ b} \\ r_1+r_2=r+r'}}\mu(P,X_1)  \sum_{\substack{f': \partial{X_1} \cup \partial{X_2}  \rightarrow  B \cup V'  \\ f' \extends f}} \ans{a_1}{r_1,X_1,f'_{|X_1 }} \cdot \ans{a_2}{r_2,X_2,f'_{|X_2}}.
	\end{align*}
	Note we sum over $r_1+r_2=r+r'$ as pattern vertices in $V_1 \cap V_2$ will contribute to the monitor count of monitors in both subproblems.
\end{proof}

We obtain the following directly as corollary:

\begin{corollary}\label{cor:sparsereduction}
	If $(V_1,V_2)$ is a separation of $G$ with $|\upp_{\cM}(V_1 \cap V_2)| \leq b$, there is a strict reduction of the subproblem $(G,B,\cM)$  to subproblems $(G[V_1],(B \cup V_1) \cap V_2,\cM[V_1])$ and $(G[V_2],(B \cup V_2) \cap V_1,\cM[V_2])$.
\end{corollary}

\paragraph{Base Case: Few Pattern Vertices}
Corollary~\ref{cor:sparsereduction} in combination with standard dynamic programming on tree decompositions gives us the following result that serves as a basis of the divide and conquer scheme we use to prove Theorem~\ref{thm:main}:
\begin{lemma}\label{lem:basecase}
	The answer to a subproblem $(G,B,\cM)$ that satisfies $\tw(G)=O(k)$, and $|V(G)|= O(\sqrt{k})$ or $|\upp_{\cM}(V(G))| = \tO(\sqrt{k})$ and $|B|=k^{O(1)}$ can be computed in $2^{\tilde{O}(\sqrt{k})}|\sepsa{\tO(\sqrt{k})}|n^{O(1)}$ time.
\end{lemma}
\begin{proof}
	The answer can be computer as follows: First compute a tree decomposition $\mathbb{T}$ of width $O(k)$ using Lemma~\ref{lem:twkout}.
	Note that any bag of $\mathbb{T}$ gives a separator of size $O(k)$ and it will include at most $\tO(\sqrt{k})$ pattern vertices since $|\upp_{\cM}(V(G))|= \tO(\sqrt{k})$ (note that we can assume $|\upp_{\cM}(V(G))| \leq n$ as values indexed by larger indices are trivially zero).
	Therefore Corollary~\ref{cor:sparsereduction} can be used in combination with any bag to reduce the subproblems.
	It remains to compute the answer of the subproblems corresponding to all leaf bags of $\mathbb{T}$, but these are of size $\tO(\sqrt{k})$, we can compute any entry of the answer of the associated subproblems by brute-forcing over all mappings.
\end{proof}


\subsection{Reductions Based on Sets of Separators}\label{subsec:sepsystemsreduction}

In this section we present three reductions that will be the main building blocks in the divide and conquer scheme we will use to prove Theorem~\ref{thm:main}.
We first present two reductions that assume a set of separators is given, which will be due to an application of Lemma~\ref{lem:chaineitherway} when we apply the lemma's.
The third reduction uses the first two reductions to obtain a cycle separator with good balance properties.

\paragraph{Set of separators.}
The following relatively direct reduction shows applies if a set of disjoint separators is given.
It may not immediately look that the subproblems are in fact easier, but this will become more clear when we apply the reduction later.
\begin{lemma}[Set of separators]\label{lem:sepsred}
	There is a strict reduction that, given a
	\begin{itemize}
		\item subproblem $(G,B,\cM)$, a triangulation $G^\Delta:=(V,E \cup \Delta)$ of $G$, and an aligned cycle $C$ of $G^\Delta$,
		\item a chain of disjoint $(C^\pointdown,C^\pointup)$-separators $(S_1,\ldots,S_p)$ with $|S_j| \leq 2q= k^{O(1)}$ for each $j \in [p]$, 
	\end{itemize}
	outputs $p$ subproblems $\{(G,B,\cM_i)\}_{i \leq p}$ such that $|\smallM(\cM_i)| \leq |\smallM(\cM)| + p$, $|\largeM(\cM_i)| = |\largeM(\cM)|$  and $\upp_{\cM_i}(V(S_i)) =  k/p$.
\end{lemma}
\begin{proof}
	We need to show how to compute $\ans{a}{r,(X,Y),f}$, where $a$ is the answer to subproblem $(G,B,\cM)$, given the answers $a_i$ to subproblems $(G,B_i,\cM_i)$.
	Note that $\ans{a}{r,(X,Y),f}$ counts mappings $g \in \ind(P[X],G)$, so $|g(X)| = |X| \leq k$. Since the $S_i$'s are disjoint, for every such $g$ there exists a unique $i$ such that $|g(X) \cap V(S_i)| \leq \sqrt{k}$ and $|g(X) \cap V(S_j)| > \sqrt{k}$ for $j < i$. We define the subproblem by setting the monitors as follows:
	\[
		\cM_i = \cM \cup \{ (V(S_j),k/p+1,k) \}_{j < i} \cup \{(V(S_i),0,k/p)\}.
	\]
	Summing over all such $i$ and counting the corresponding functions $g$ we see that
	\begin{equation}\label{eq:setofseps}
		\ans{a}{r,(X,Y),f} = \sum_{i=1}^p \sum_{r^i_{V(S_1)},\ldots,r^{i}_{V(S_i)}} \ans{a_i}{r^i,(X,Y),f},
	\end{equation}
	where in the second sum $r^i_{V(S_j)}$ ranges over $\{k/p+1,\ldots,k\}$ if $j < i$ and over $\{0,\ldots,k/p\}$ if $i=j$.
	
	This gives a strict reduction since all values of the answer can be computed in time linear in the size of the output answer table by straightforward evaluation of~\eqref{eq:setofseps}, and the number of non-boundary vertices and pattern vertices in the subproblems is not larger than in $(G,B,\cM))$.
\end{proof}

\paragraph{Set of nearly disjoint paths.}
The setting of the following reduction is similar to that of the previous reduction, but the reduction itself is harder.
The reason is that, since we will use $q \geq k$, we cannot afford to add monitors for every $S_i$.
To circumvent adding all these monitors, we apply efficient inclusion-exclusion in a way similar to Lemma~\ref{lem:redkout}.
But to use this strategy, we require a chain of separators in $G$ and the given separators $S_1,\ldots,S_q$ (which can also be viewed as paths in $G^\Delta$) only are separators in the graph $G[\int(C)]$.
To resolve this, we consider two vertices $s,t$ that are on sufficiently many of the paths and pair-wise combine the paths to obtain a chain of nearly-disjoint cycles.
This brings us to the same setting as in Lemma~\ref{lem:redkout}, except that the cycles may overlap in few vertices.
To deal with this overlap, we sum over all mappings of the pattern to these boundary vertices; this can be done implicitly by including them as indices in the dynamic programming tables.

\begin{lemma}[Set of nearly disjoint paths]\label{lem:dispathred}
	There is a strict reduction that, given a
	\begin{itemize}
		\item subproblem $(G,B,\cM)$, a triangulation $G^\Delta:=(V,E \cup \Delta)$ of $G$, and an aligned cycle $C$ of $G^\Delta$,
		\item system of $(C^{\pointleft},C^{\pointright})$-separators $(S_1,\ldots,S_q)$ of $G^\Delta$ with $|\pub(S_j)| \leq 4\sqrt{k}$ for $j \in [q]$ where $q\geq 4k^3$,
	\end{itemize}
	outputs a set $\cI \cup \cO$ of $O(q^2)$ subproblems $\{(G_l,B_l,\cM_l)\}_{l}$ with associated separator $S_l$ such that
	\begin{itemize}
		\item $|B_l \setminus B|\leq 16\sqrt{k}$,
		\item $\upp_{\cM_l}(B_i)\leq |\upp_{\cM}(B)|+ 16\sqrt{k}$,
		\item $\cM_l = \cM[V(G_l)] \cup (\priv(S_l)\cup \priv(S_{j'}),0,0) \text{ for some } j' \in [q]$.
	\end{itemize}
	Moreover, $G_l=G$ if $(G_l,B_l,\cM_l) \in \cI$, and $G_l$ is a subgraph of $\int(C)$ if $(G_l,B_l,\cM_l) \in \cO$ .
\end{lemma} 
\begin{proof}
	Note that since $G^\Delta$ is a triangulated graph, each $(C^\pointleft,C^\pointright)$-separator $S_i$ is a path between a vertex from $C^\pointdown$ and a vertex from $C^\pointup$.
	By the pigeon-hole principle we can find two vertices $s \in C^{\pointdown}$ and $t \in C^{\pointup}$ such that there are at least $k+1$ paths $S_i$ that contain both $s$ and $t$.
	Let $S'_1,\ldots,S'_{2k+2}$ be a subsequence of $S_1,\ldots,S_q$ that consists of only paths from $s$ to $t$.
	Note that if we let $C_i = S_i \cup S_{2k+3-i}$, then $C_1,\ldots,C_{k+1}$ forms a chain of nested cycles, i.e. $\int(C_i) \subseteq \int(C_{i-1})$ for every $i$.
	
	Define $\pub(C_i)=\pub(S_i) \cup \pub(S_{2k+3-i})$ and $\priv(C_i)=\priv(S_i) \cup \priv(S_{2k+3-i})$. Since $g(X) \cap \priv(C_i) = \emptyset$ for some $i$ we have by~\eqref{eq:ie} that
	\begin{align}\label{eq:bigie}
		a[r,(X,Y),f] &= \Big|\bigcup_{i \in [k+1]}\Big\{ g\in \ind(P[X],G - \priv(C_i)): g \extends f \wedge g \picks r \Big\}\Big| \nonumber\\
					 &= \sum_{\emptyset = R \subseteq [k+1]}(-1)^{|R|-1} a_R,\\
\text{where} \quad a_R &= \Big|\Big\{ g\in \ind\Big(P[X],G - \bigcup_{i \in R} \priv(C_i)\Big):  g \extends f \wedge g \picks r \Big\}\Big|.\nonumber
	\end{align}
	For integers $i<j \in [k+1]$ we define $G_{ij}$ as $G[(\int(C_i) \setminus \int(C_j))\cup \pub(C_j)]$. Let $R= \{z_1 < \ldots < z_h\}$. Applying Lemma~\ref{lem:sparsesepform} for separators $\pub(S_{r_i})$ in the graph $G - \bigcup_{i=1}^h \priv(S_{r_i})$, we claim that $a_R$ equals
	\begin{equation}\label{eq:bigformula1}
		\begin{aligned}
		\sum_{
			\substack{
				X_0 \subseteq \ldots \subseteq X_h = X \\
				|\partial X_i|\leq |\pub(C_i)| \\
				f^1,\ldots,f^h\\
				r^1+\ldots+r^{h+1}=r+r'
			}
		 }
	 &\Big|\Big\{g_0 \in \ind\Big(P[X_0],\ext(C_{z_1})-\priv(C_{z_1})\Big): g_0 \extends f^1 \wedge g_0 \picks r^1 \Big\}\Big|*\\
	 &\prod_{i=1}^{h-1} \Big|\Big\{g_i \in \ind\Big(P[X_i\setminus X_{i-1}],G_{z_i,z_{i+1}}\Big): g_i \extends f^i \text{ and } f^{i+1} \wedge g_i \picks r^i \Big\}\Big|.
	 	\end{aligned}
	\end{equation}
	Here the sum ranges over all functions $f^1,\ldots,f^h$ where $f^i: X \rightarrow \pub(C_i)$ that agree with $f$, and 
	all $r^1,\ldots,r^{h+1}$ such that $r^1+\ldots+r^{h+1}=r+r'$, where $r'_M=\sum_{v \in M}c_v$ and $c_v$ is the number of $i$ such that $v \in X_i$.

	To see this more directly than to apply Lemma~\ref{lem:sparsesepform}, note we can group the set of all possible functions $g$ in $\ind(P[X],G-\cup_{i \in R}\priv(C_i))$ that contribute to $a_R$ based on their image $I$ by summing over all $Y_0,\ldots,Y_h$ where $Y_0 = I \cap V(\ext(C_{z_1}))$ and $Y_i= I \cap V(G_{z_i,z_{i+1}})$ for $i>0$.
	This uniquely defines $X_j=\cup_{i \leq j}Y_j$. 
	Thus we may count the functions by summing over all such $X_i$.
	Note that if $X_i \setminus X_{i-1}$ is the set of vertices mapped by $g$ to $G_{z_i,z_{i+1}}$ then $g(\partial (X_i)) \subseteq \pub(C_i)$ as $G_{z_i,z_{i+1}}$ is a connected component of the graph $G-\cup_{i \in R}\priv(C_i)$.
	Thus, if we subsequently group the mapping $g$ on its projected mappings $f^1,\ldots,f^h$ we can count the remaining functions $g$ as a product of mappings in the graphs $\ext(C_{z_1})$ and $G_{z_i,z_{i+1}}$ separately.

	Now we combine (\ref{eq:bigie}) and (\ref{eq:bigformula1}) to get that $\ans{a}{r,(X,Y),f}$ equals $\sum_{h=1}^{k+1}T_h$ where $T_h$ equals
	\begin{equation}\label{eq:biggerformula1}
		\begin{aligned}
		\sum_{z_1 < \ldots < z_h}&
	\sum_{
		\substack{
			X_0 \subseteq \ldots \subseteq X_h = X \\
			|\partial X_i|\leq |\pub(C_{z_i})| \\
			f^1,\ldots,f^h\\
			r^1+\ldots+r^{h+1}=r+r'
		}
	}
	\Big|\Big\{g_0 \in \ind\Big(P[X_0],\ext(C_{z_1}-\priv(C_{z_1}))\Big): g_0 \extends f^1 \wedge g_0 \picks r^1 \Big\}\Big|*\\
	&\prod_{i=1}^{h-1} \Big|\Big\{g_i \in \ind\Big(P[X_i\setminus X_{i-1}],G_{z_i,z_{i+1}}\Big): g_i \extends f^i \text{ and } f^{i+1} \wedge g_i \picks r^i \Big\}\Big|.
		\end{aligned}
	\end{equation}
	Now we can see that \eqref{eq:biggerformula1} equals~\eqref{eq:sumproductwrap} when we substitute $A$ with $\sepsa{\max_i |\pub(C_i)|} \times [k+1]$ and $T[((X_1,Y_1),r^1,f^1,z),((X_2,Y_2),r^2,f^2,z')]$ with $\mult{((X_1,Y_2),X_2)}\cdot X$ where
	\[Q =
	\begin{cases}
		 \big|\big\{ g \in \ind(P[X_i \setminus X_{i-1}],G'_{z,z'}): g \extends f^1,f^2 \text{ and } g \picks r^1+r^2\big\}\big|, & \text{if } z < z'\\
		 \big|\big\{ g \in \ind(P[X_i \setminus X_{i-1}],\ext(C_{z'})-\priv(C_{z_1})): g \extends f^2 \text{ and } g \picks r^2\big\}\big|, & \text{if } z > z' \text{ and }f^1=r^1=\emptyset,\\
		 0,& \text{otherwise}.
	\end{cases}
	\]
	Therefore, by Corollary~\ref{lem:techDPwrap} we can evaluate~\eqref{eq:biggerformula1} and thus $\ans{a}{r,(X,Y),f}$ if the value of $Q$ for all relevant parameters is given.
	Thus the values of $Q$ can be read off from the answers of the subproblems

	\begin{align}
		\nonumber&(G_{ij},&&(B \cap V(G_{ij})) \cup \pub(C_i) \cup \pub(C_j),&&\cM[V(G_{ij})] )\\
		\nonumber&(\ext(C_{z'}),&&(B \cap V(\ext(C_{z'}))\setminus \priv(C_{z_1})) \cup \pub(C_i) \cup \pub(C_j),&&\cM[V(\ext(C_{z'}))])
	\end{align}

	It remains to show that these subproblems satisfy the properties of the lemma statement.
	Indeed $B$ increases by at most $2|\pub(C_j)| \leq 4|\pub(S_j)| \leq 16\sqrt{k}$ vertices in every subproblem.
	It is easy to verify that the subproblems satisfy the required properties stated in the Lemma, and they can be computed in time polynomial in the input/output by Corollary~\ref{lem:techDPwrap}.
\end{proof}

\paragraph{Acquiring Balance}
Now we show how to obtain a short cycle separator with good balance properties.
In particular, we use a win-win approach:
We consider a cycle $C$ on at most $k$ vertices, and distinguish the following two cases. Either only $8\sqrt{k}$ vertices of the pattern are contained in $V(C)$, in which case we can use $C$ as a good separator (as the number of mapping of the pattern to $C$ will be $2^{\tO(\sqrt{k})}$), or we have the rough location of at least $\sqrt{k}$ vertices.
We show below the latter can be leveraged via Lemma~\ref{lem:chaineitherway} and the previous lemma's to reduce the computation of the subproblem at hand to a subproblem that amounts to count patterns for which another cycle separator $C'$ separates at least $\sqrt{k}$ vertices of the pattern.

\begin{lemma}[Acquiring Balance]\label{lem:acqbal}
	Fix $\theta=100\sqrt{k}\log_{4/3}(k)$. There is a strict reduction from subproblem $(G,B,\cM)$ to $O(k^6)$ subproblems $\{(G_i,B_i,\cM_i)\}_{i \leq l}$ with the following properties: For $i \in \{1,2\}$ we have 
	\begin{equation}\label{eq:easysp}
			|V(G_i) \setminus B_i| \leq \tfrac{3}{4}|V(G)\setminus B|,\ \ \cM_i = \cM[V(G_i)], \ \ |B_i| = |B|+ O(k^3), \ \ 
			\upp_{\cM_i}(B_i) \leq |\upp_\cM(B)|+\theta
	\end{equation}
	and for $i>2$ we have that
	\begin{itemize}
		\item $\upp_{\cM_l}(B_i)\leq |\upp_{\cM}(B)|+ 16\sqrt{k}$,
		\item $|B_i \setminus B| \leq 16\sqrt{k}$,
		\item $|\smallM(\cM_i)| \leq |\smallM(\cM)|+4+\sqrt{k}$,
		\item $|\largeM(\cM_i)| \leq 4+|\largeM(\cM)|$.
	\end{itemize}
	Moreover, for $i>2$ either $\upp_{\cM_i}(V(G_i))\leq \upp_{\cM}(V(G))-\sqrt{k}$ or the reduction outputs an associated cycle $C_i$ with the property that $\low_{\cM_i}(\int_{G_i}(C_i)), \low_{\cM_i}(\ext_{G_i}(C_i)) \geq \theta/4$.
\end{lemma}
\begin{algorithm}[h]
	\caption{Reduction to obtain a small cycle with good balance properties from Lemma~\ref{lem:acqbal}}
	\label{alg:acquireBalance}
	\begin{algorithmic}[1]
		\REQUIRE $\texttt{acquireBalance}(G,B,\cM,\Delta)$ \hfill\algcomment{Assumes $G^\Delta=G + \Delta$ is triangulated and $k$-outerplanar}
		\ENSURE Subproblems $\{(G_i,B_i,\cM_i)\}_{i \leq l}$ and associated cycles $C_i$ as stated in Lemma~\ref{lem:acqbal}
		\STATE Find a cycle $C$ on at most $k$ vertices that is $\tfrac{3}{4}$-balanced for $V(G) \setminus B$ with Lemma~\ref{lem:balcyc}
		\STATE Use Corollary~\ref{cor:sparsereduction} on subproblem $(G,B,\cM \cup \{V(C),0,\theta\})$ and separation $(\int_{G^\Delta}(C),\ext_{G^\Delta}(C))$\label{lin:redsparse}
		\STATE Let $res$ be the set with the resulting two subproblems
		\FOR{every alignment $\{C^\pointleft,C^\pointdown,C^\pointright,C^\pointup\}$ of $C$}\label{lin:foralign}
		\STATE $\cM \gets \cM \cup \{ (C^\pointleft, \theta/4,\theta/4),(C^\pointdown, \theta/4,\theta/4),(C^\pointright, \theta/4,\theta/4),(C^\pointup, \theta/4+1,\infty)\}$
		\STATE $\cM_{\int_G(C)} \gets \{(V(\int_G(C)),\theta/2,k)\}$\label{lin:acqmonint}
		\STATE $\cM_{\ext_G(C)} \gets \{(V(\int_G(C)),0,\theta/2-1), (V(\ext_G(C)),\theta/2,k)\}$
			\FOR{$\side \in \{\int_G(C),\ext_G(C)\}$}\label{lin:chooseside}
				\STATE $\cM \gets \cM \cup \cM_\side$
				\STATE Apply $\mathtt{menger+}$ (Lemma~\ref{lem:chaineitherway}) with $\side(G^\Delta)$, aligned cycle $C$, $p=\sqrt{k}$, and $q=3k^3$
				\STATE Let $S_1,\ldots,S_r$ be the obtained output\label{lin:getseps}
				\IF[$S_1,\ldots,S_r$ are disjoint separators]{$r=p$}
				\STATE Use Lemma~\ref{lem:sepsred} on $(G,B,\cM)$, $G^\Delta$, $C$ and $S_1,\ldots,S_r$. Add obtained subproblems to $\cA$.
				\ELSE[$S_1,\ldots,S_r$ are nearly disjoint paths]
				\STATE Use Lemma~\ref{lem:dispathred} on $(G,B,\cM)$, $G^\Delta$, $C$ and $S_1,\ldots,S_r$; let $\cI \cup \cO$ be the output
				\STATE Add all obtained subproblems in $\cI$ to $res$, and all subproblems in $\cO$ to $\cA$.
				\ENDIF
				\FOR{every subproblem $(G_i,B_i,\cM_i) \in \cA$ with $S_i$ associated separator $S_i$}\label{lin:associate}
				\STATE Consider the cycles $C^{\circlearrowleft}$ and $C^{\circlearrowright}$ of the graph $C \cup S_i$ that contain $S_i$\label{lin:buildcyc}
				\STATE Add subproblem $(G_i,B_i,\cM_i \cup \{(\int_G(C^{\circlearrowleft}),\theta/4,k) \})$ to $res$\label{lin:largecycmon1}
				\STATE Add subproblem $(G_i,B_i,\cM_i \cup \{(\int_G(C^{\circlearrowleft}),0,\theta/4-1),(\int_G(C^{\circlearrowright}),\theta/4,k) \})$ to $res$\label{lin:largecycmon2}
				\STATE Output as associated cycles of these subproblems respectively $C^{\circlearrowleft}$ and $C^{\circlearrowright}$\label{lin:acqoutput}
				\ENDFOR
			\ENDFOR
		\ENDFOR
		\STATE \algorithmicreturn\ $res$
	\end{algorithmic}
\end{algorithm}
\begin{proof}
The promised reduction is implemented in Algorithm~\ref{alg:acquireBalance}.
Let $g$ be a function contributing to a value of an answer, and let $X$ be the image of $g$.
The algorithm distinguishes functions $g$ with at most, and respectively more than, $\theta$ vertices in $V(C) \cap X$ separately.
At Line~\ref{lin:redsparse} we count all pattern occurrences with at most $\theta$ such vertices: We count these by adding two subproblems created by Corollary~\ref{cor:sparsereduction} to the output and using the reduction from Corollary~\ref{cor:sparsereduction} to compute the associated values.

It remains to count pattern occurrences $g$ with more than $\theta$ vertices in $V(C) \cap X$.
Note that for each such pattern there is a unique alignment $\{C^\pointleft,C^\pointdown,C^\pointright,C^\pointup\}$ such that $|C^\pointleft \cap X|=|C^\pointdown \cap X|=|C^\pointright \cap X|=\theta/4$, and $|C^\pointup \cap X|> \theta/4$: Since the ordering $v_1,\ldots,v_l$ is fixed, $C^\pointleft$ must be equal to $\{v_1,\ldots,v_j\}$ where $j$ is the $\theta/4$-smallest element of $X$ in the ordering (and similar arguments fixed $C^\pointdown, C^\pointright$ and $C^\pointup$). Moreover, if $|V(C) \cap X| \leq \theta$ such an alignment does not exist.
In the corresponding iteration of the loop at Line~\ref{lin:foralign} we will show the algorithm counts the corresponding matching occurrences.

At Line~\ref{lin:chooseside} we choose to continue with either the interior or the exterior of the cycle $C$.
The large monitor $\cM_\side$ ensures that the chosen side contains at least $\theta/2$ vertices of $X$.
Note that every vector feasible for $\cM$ is feasible for exactly one of the monitor families $\cM_{\int_G(C)}$ and $\cM_{\ext_G(C)}$ and for this choice the algorithm will count the corresponding functions $g$.

Now we obtain a set of separators in the restricted graph at Line~\ref{lin:getseps} that we use to create subproblems according to Lemma~\ref{lem:sepsred} and Lemma~\ref{lem:chaineitherway}.
Note the graph $C \cup S_i$ is a cycle together with a path between two vertices on the cycle. In this graph $S_i$ is contained in two simple cycles that we call $C^{\circlearrowleft}$ and $C^{\circlearrowright}$ on Line~\ref{lin:buildcyc}.
Since $V(\int_{G^\Delta}(C^{\circlearrowleft})) \cup V(\int_{G^\Delta}(C^{\circlearrowright}))=\int_{G^\Delta}(C)$ we have that the interior of at least one of $C^{\circlearrowleft}$ and $C^{\circlearrowright}$ contains at least $\theta/4$ pattern vertices. Therefore we can add the large monitors on Line~\ref{lin:largecycmon1} and Line~\ref{lin:largecycmon2} and still ensure to count every function exactly once.

Therefore, the reduction is correct, and it can be easily verified that all created subproblems have the properties stated in the lemma, and that the reduction is strict.
\end{proof}
\section{Main Reductions and Algorithm}\label{sec:mainred}

With all tools from the previous section set up, we are ready to present our two main subroutines.
On a high level, our main algorithm (called \texttt{solveSubproblem} in Algorithm~\ref{alg:main}) is a divide and conquer scheme that may invoke either of the following two reductions:

\paragraph{The main reduction}\hspace{-1em} starts with a cycle separator as given by Lemma~\ref{lem:acqbal}.
Its goal is to use this to reduce the subproblem to two simpler subproblems as in Lemma~\ref{lem:basecase}.
In order to do so, we show how to find a family of cycle separators with equally good balance properties with the property that at least one of them must have few vertices of the pattern (which we informally call `sparsifying balanced separators' in the introduction). The construction of this family is based on Lemma~\ref{lem:chaineitherway} in a way very similar way to the proof of Lemma~\ref{lem:acqbal}.

\paragraph{The clean-up reduction}\hspace{-1em} is to ensure we only create subproblems $(G,B,\cM)$ satisfying $\upp_{\cM}(B)=\tO(\sqrt{k})$, $|B| = k^{O(1)}$ and 
$|\smallM(\cM)|,|\largeM(\cM)| \leq \tO(\sqrt{k})+ O(\log n / \log k)$.
Note that this is sufficient to guarantee that all the sizes of the answers are $2^{\tO(\sqrt{k})}n^{O(1)}|\sepsa{\tO(\sqrt{k})}|$.

To ensure this property on all created subproblems we choose our initial balanced cycle separator to be balanced on the set $W=B \cup_{M \in \smallM(\cM)}M$ which we sparsify in a way similar to as in the main reduction. Since $|W|$ is polynomial in $k$, we can repeat this $O(\log k)$ steps to ensure each created subproblem has only $\tO(\sqrt{k})$ vertices of $W$, which suffices to ensure the subproblem satisfies all above guarantees.
	
Large monitors cannot be removed in this way as their size may be unbounded in terms of $k$, but fortunately this is not a problem since we only add relatively few of them.
In Subsection~\ref{subsec:mainalgo} we will discuss this in more detail when we combine the two above reductions.

\subsection{Main Reduction: Decreasing the Number of (Pattern) Vertices}
Now we outline the most important reduction.
On a high level we give a reduction that decreases at least one of the two most important quantifications of the complexity of a subproblem, the size of the host graph (i.e~$|V(G)|$), and the number of pattern vertices that still need to be mapped (i.e. $\upp_{\cM}(V(G))$). The main reduction produces 2 subproblems in which $|V(G)|$ decreases with a constant fraction, and $k^{O(\log k)}$ subproblem in which $\upp_{\cM}(V(G))$ decreases with at least $\Omega(\sqrt{k})$.

In the algorithm the following notation will be useful:
\begin{definition}[Annotated Cycle]
An \emph{annotated cycle} is a cycle $C$ along with a partition of $V(C)$ into $C_\heavy,C_\light$ and $C_\discard$.
\end{definition}
The underlying meaning of $C_\heavy,C_\light$ and $C_\discard$ are that they denote a set of `heavy' vertices (from which any subset could occur in the pattern), `light' vertices (from which only $\tO(\sqrt{k})$ of such vertices can be selected), and `discarded' vertices (from which no such vertices can be selected). In the following we slightly abuse notation by referring to $C$ as the cycle together with the partition.
\begin{lemma}[Main Reduction]\label{lem:mainred}
	There is a strict reduction from the subproblem $(G,B,\cM)$ to $l = k^{O(\log k)}$ subproblems $\{(G_i,B_i,\cM_i)\}_{i \leq l}$ with the following properties: For $i \in \{1,2\}$ we have
	\begin{align*}
	|V(G_i) \setminus B_i| \leq \tfrac{3}{4}|V(G) \setminus B|,\qquad  &\cM_i = \cM[V(G_i)], \\
	|B_i| = |B|+ O(k^3)\qquad &\upp_{\cM_i}(B_i) \leq |\upp_\cM(B)|+\theta:=100\sqrt{k}\log k,
	\end{align*}
	and for $i > 2$ we have	
	\begin{align*}
	 	|\largeM(\cM_i)| &\leq |\largeM(\cM)| + O(\log k),     &|\smallM(\cM_i)| \leq |\smallM(\cM)| + O(\sqrt{k}\log k),\\
		\upp_{\cM_i}(B_i) &\leq \upp_{\cM}(B)+O(\sqrt{k} \log k), &\upp_{\cM_i}(V(G_i)) \leq \upp_{\cM}(V(G)) - \Omega(\sqrt{k}),\\
		|B_i| &\leq |B|+O(k^3 \log k).
	\end{align*}
\end{lemma}
\begin{algorithm}[h]
	\caption{The main reduction from Lemma~\ref{lem:mainred}}
	\label{alg:mainreduction1}
	\begin{algorithmic}[1]
		\REQUIRE $\texttt{mainReduce}(G,B,\cM)$ \hfill\algcomment{Assumes $G$ to be $k$-outerplanar}
		\ENSURE $k^{O(log k)}$ subproblems $\{(G_i,B_i,\cM_i)\}_{i \leq l}$ with the properties of Lemma~\ref{lem:mainred}
		\STATE Triangulate $G$ with Lemma~\ref{lem:outerplanar} so that $G^\Delta:=G\cup \Delta$ is triangular and $(k+1)$-outerplanar
		\STATE $\cS \gets \mathtt{acquireBalance}(G,B,\cM,\Delta)$\label{lin:invokeab}
		\STATE $res \gets  \Big\{ \pi_i=(G_i,B_i,\cM_i) \in \cS : \pi_i \text{ satisfies } \eqref{eq:easysp} \text{ or } \upp_{\cM_i}(V(G_i))\leq \upp_{\cM}(V(G))-\sqrt{k} \Big\} $ 
		\FOR{$(G,B,\cM) \in \cS \setminus res$ with associated cycle $C$}\label{lin:mainconsidercycle}
		\STATE Let $C=(C_\heavy, C_\light, C_\discard) := (V(C), \emptyset, \emptyset)$ \hfill\algcomment{$C$ denotes the annotated cycle} \label{lin:initialannotation}
		\STATE Let $\cA^0 := \{((G,B,\cM),C)\}$ \label{lin:initialcycle}
		\FOR{$d=1,\ldots,z:= \log_{4/3} k$}\label{lin:loopsparsify}
		\STATE Initiate $\cA^d := \emptyset$
		\FOR{$((G,B,\cM),C) \in \cA^{d-1}$}\label{lin:pickcycl}
		\STATE Find alignment $C^\pointleft,C^\pointdown,C^\pointright,C^\pointup$ of $C$ with $\forall x\in \{\pointleft,\pointdown,\pointright,\pointup\}: |C^x \cap C_{\heavy}| \geq \lfloor |C_{\heavy}| / 4 \rfloor$  
		\STATE $\cM_{\int_G(C)} \gets \{(V(\int_G(C)),4\sqrt{k},k)\}$\label{lin:mainmonint}
		\STATE $\cM_{\ext_G(C)} \gets \{(V(\int_G(C)),0,4\sqrt{k}-1), (V(\ext_G(C)),4\sqrt{k},k)\}$
		\FOR{$\side_G(C) \in \{\int_G(C),\ext_G(C)\}$}
		\STATE $\cM \gets \cM \cup \cM_\side$
		\STATE Apply $\mathtt{menger+}$ (Lemma~\ref{lem:chaineitherway}) with $\side(G')$, aligned cycle $C$, $p=\sqrt{k}$, and $q=3k^3$
		\STATE Let $S_1,\ldots,S_r$ be the obtained output
		\IF[$S_1,\ldots,S_r$ are disjoint separators]{$r=p$}
		\STATE Use Lemma~\ref{lem:sepsred} on $(G,B,\cM)$, $G^\Delta$, $C$ and $S_1,\ldots,S_r$. Add obtained subproblems to $\cB$.
		\ELSE[$S_1,\ldots,S_r$ are nearly disjoint paths]
		\STATE Use Lemma~\ref{lem:dispathred} on $(G,B,\cM)$, $G^\Delta$, $C$ and $S_1,\ldots,S_r$; let $\cI \cup \cO$ be the output
		\STATE Add all obtained subproblems in $\cI$ to $res$, and all subproblems in $\cO$ to $\cB$.
		\ENDIF
		\FOR{every subproblem $(G_i,B_i,\cM_i) \in \cB$ with $S_i$ associated separator $S_i$}
		\STATE Consider the cycles $C^{\circlearrowleft}$ and $C^{\circlearrowright}$ of the graph $C \cup S_i$ that contain $S_i$
		\item[] \hfill\algcomment{We let $\ann(\cdot)$ annotate a cycle according to~\eqref{eq:anncyc}}
		\STATE Add $\big( (G_i,B_i,\cM_i \cup \{(\int_G(C^{\circlearrowleft}),2\sqrt{k},k) \}), \ann(C^{\circlearrowleft}) \big)$ to $\cA^d$\label{lin:mainoutputcycss}
		\STATE Add $\big( (G_i,B_i,\cM_i \cup \{(\int_G(C^{\circlearrowleft}),0,2\sqrt{k}-1),(\int_G(C^{\circlearrowright}),2\sqrt{k},k)\}), \ann(C^{\circlearrowright}) \big)$ to $\cA^d$ \label{lin:mainoutputcycs}
		\ENDFOR
		\ENDFOR
		\ENDFOR
		\ENDFOR
		\FOR{$( (G_i,B_i,\cM_i), C) \in \cA^z$}
		\STATE Apply Corollary~\ref{cor:sparsereduction} on $(G_i,B_i,\cM_i)$ and separator $C_\heavy \cup C_\light$
		\STATE Add the resulting two subproblems to $res$\label{lin:mainaddtores}\hfill\algcomment{Uses that $C_\heavy \cup C_\light$ is a sparse separator}
		\ENDFOR
		\ENDFOR
		\RETURN $res$
	\end{algorithmic}
\end{algorithm}
The reduction from Lemma~\ref{lem:mainred} is implemented in Algorithm~\ref{alg:mainreduction1} (note Lines~\ref{lin:mainmonint}-\ref{lin:mainoutputcycs} are similar to Lines~\ref{lin:acqmonint}-\ref{lin:acqoutput} of Algorithm~\ref{alg:acquireBalance}).
We complete the description of the algorithm by defining the annotation of the produced cycles as follows:
\begin{equation}\label{eq:anncyc}
\ann(C^\alpha) :=
\begin{cases}
\big(C_\heavy \cap V(C^\alpha),C_\light \cap V(C^\alpha) \cup S_i,C_\discard \cap V(C^\alpha)\big), & \text{ if $r=p$}\\
\big(C_\heavy \cap V(C^\alpha) \cup \pub(S_i),C_\light \cap V(C^\alpha),C_\discard \cap V(C^\alpha) \cup \priv(S_i)\big), & \text{ otherwise}.
\end{cases}
\end{equation}
Note this is consistent with the monitors produced by the reductions from Lemma~\ref{lem:sepsred} and Lemma~\ref{lem:dispathred}.	
\begin{proof}
	We may only compute answer values corresponding to functions that map $\Omega(\sqrt{k})$ pattern vertices to $G$ as otherwise the other functions can be computed using Lemma~\ref{lem:basecase}.
	
	At Line~\ref{lin:invokeab} we obtain a set $\cS$ of subproblems. By Lemma~\ref{lem:acqbal} this set contains $2$ problems $(G_i,B_i,\cM_i)$ that satisfy $|V(G_i)|\leq \tfrac{3}{4}|V(G)|+k$, and $O(k^6)$ subproblems that satisfy $\upp_{\cM_i}(V(G_i))\leq \upp_{\cM}(V(G))-\sqrt{k}$. As the first two subproblems meet the criteria stated in the lemma, we directly add them to the output of the reduction and don't consider them anymore.
	The last category of subproblem in $\cS$ come with an associated cycle $C_i$ with the property that $\low_{\cM_i}(\int_{G_i}(C_i)), \low_{\cM_i}(\ext_{G_i}(C_i)) \geq \theta/4$. These subproblems cannot be added directly to the output.
	
	Therefore, we consider each of these subproblems and their associated cycles in Lines~\ref{lin:mainconsidercycle}-\ref{lin:mainaddtores}, and reduce the subproblems to subproblems with different associated cycles with equally good balance properties but with the guarantee that they contain few vertices of the pattern.
	Then we can use Corollary~\ref{cor:sparsereduction} to reduce each subproblem to two subproblems that have the requirements of the lemma statement by the balance properties of the cycles.
	
	To do so, we maintain a set of pairs $\cA^d$ consisting of subproblems and annotated cycles, such that the subproblem contains monitors that ensure few vertices from $C_\light$ and no vertices from $C_\discard$ are included in the pattern.
	Initially, the associated cycle $C$ only consists of heavy vertices.
	Note that in any of the $d$ iterations of the loop at Line~\ref{lin:loopsparsify}, $C_\heavy$ is decreased with a multiplicative factor $3/4$ if it is of size $\Omega(\sqrt{k})$, since one of the four parts of the alignment that contains $|C_\heavy|/4$ heavy vertices is replaced with $C_i$ that contains at most $4\sqrt{k}$ heavy vertices. Moreover, we only add $O(k^3)$ vertices to $|C_\light|$ per iteration of the loop at Line~\ref{lin:loopsparsify} and thus $O(k^3\log k)$ in total.
	Thus indeed $C$ is a sparse separator at Line~\ref{lin:mainaddtores} and Corollary~\ref{cor:sparsereduction} applies.

	Furthermore, we claim that all associated cycles maintain the property that at least $\theta/4$ pattern vertices are in the strict interior and strict exterior of the cycle. Let us call this the balance property.	
	To see this, consider some cycle $C$ at Line~\ref{lin:mainconsidercycle} and assume it has the balance property. 
	
	Consider the cycles $C^\alpha$, $\alpha \in \{\circlearrowleft,\circlearrowright\}$ we add to $\cA^d$ on Lines~\ref{lin:mainoutputcycss} and~\ref{lin:mainoutputcycs}.
	If $\side=\int_G(C)$, then $\int_G(C^\alpha)\subseteq \int_G(C)$, and thus $\ext_G(C^\alpha)$ contains at least $\theta/4$ pattern vertices as $\ext_G(C)$ does so.
	If $\side=\ext_G(C)$, then $\int_G(C) \subseteq \ext_G(C^\alpha)$, and thus $\ext_G(C^\alpha)$ contains at least $2\sqrt{k}$ pattern vertices as $\int_G(C)$ does so.
	Since at Line~\ref{lin:mainapplysparse} $C_\heavy \cup C_\light$ has only $16\sqrt{k}\log_{4/3}k$ pattern vertices, both the strict interior and the strict exterior of $C$ must also contain $\theta/2-16\sqrt{k}\log_{4/3}k=\Omega(\sqrt{k})$ pattern vertices.
	
	Therefore, the reduction is correct, and it can be easily verified that all created subproblems have the remaining properties stated in the lemma.
\end{proof}

\subsection{Clean-Up Reduction: Shrinking The Boundary and Small Monitored Sets}
\begin{algorithm}[h]
	\caption{Clean-Up Step reduction from Lemma~\ref{lem:cleanstep}}
	\label{alg:halvewaste}
	\begin{algorithmic}[1]
		\REQUIRE $\mathtt{cleanStep}(G,B,\cM,W)$\hfill\algcomment{Assumes $W \subseteq V(G)$, $|W|=\poly(k)$ and $G$ is $k$-outerplanar}
		\ENSURE $k^{O(\log k)}$ subproblems with the properties of Lemma~\ref{lem:cleanstep}
		\STATE Triangulate $G$ with Lemma~\ref{lem:outerplanar} so that $G^\Delta:=G\cup \Delta$ is triangular and $(k+1)$-outerplanar
		\STATE Find a separator $C$ on at most $k$ vertices balanced for vertices in $W$ using Lemma~\ref{lem:balcyc}\label{lin:balW}
		\STATE Let $C=(C_\heavy, C_\light, C_\discard) := (V(C), \emptyset, \emptyset)$ \hfill\algcomment{$C$ denotes the annotated cycle} \label{lin:initialannotation2}
		\STATE Let $\cA^0 := \{((G,B,\cM),C)\}$ \label{lin:initialcycle2}
		\FOR{$d=1,\ldots,z:= \log_{4/3} k$}\label{lin:loopsparsify2}
		\STATE Initiate $\cA^d := \emptyset$
		\FOR{$((G,B,\cM),C) \in \cA^{d-1}$}\label{lin:pickcycl2}
		\STATE Find alignment $C^\pointleft,C^\pointdown,C^\pointright,C^\pointup$ of $C$ with $\forall x\in \{\pointleft,\pointdown,\pointright,\pointup\}: |C^x \cap C_{\heavy}| \geq \lfloor |C_{\heavy}| / 4 \rfloor$  
		\LineIfElse{$|V(\int_G(C)) \cap W|\geq |V(\ext_G(C)) \cap W|$}{$\side \gets \int_G(C)$}{$\side \gets \ext_G(C)$}
		\STATE Apply $\mathtt{menger+}$ (Lemma~\ref{lem:chaineitherway}) with $\side(G')$, aligned cycle $C$, $p=\sqrt{k}$, and $q=3k^3$
		\STATE Let $S_1,\ldots,S_r$ be the obtained output
		\IF[$S_1,\ldots,S_r$ are disjoint separators]{$r=p$}
		\STATE Use Lemma~\ref{lem:sepsred} on $(G,B,\cM)$, $G^\Delta$, $C$ and $S_1,\ldots,S_r$. Add obtained subproblems to $\cB$.
		\ELSE[$S_1,\ldots,S_r$ are nearly disjoint paths]
		\STATE Use Lemma~\ref{lem:dispathred} on $(G,B,\cM)$, $G^\Delta$, $C$ and $S_1,\ldots,S_r$; let $\cI \cup \cO$ be the output
		\STATE Add all obtained subproblems in $\cI$ to $res$, and all subproblems in $\cO$ to $\cB$.
		\ENDIF
		\FOR{every subproblem $(G_i,B_i,\cM_i) \in \cB$ with $S_i$ associated separator $S_i$}
		\STATE Consider the cycles $C^{\circlearrowleft}$ and $C^{\circlearrowright}$ of the graph $C \cup S_i$ that contain $S_i$
		\STATE Assume $|V(\sext_{G^\circlearrowleft}(C)) \cap W|\leq |V(\sext_{G^\circlearrowright}(C)) \cap W|$ by relabeling, if needed
		\item[] \hfill\algcomment{We let $\ann(\cdot)$ annotate a cycle according to~\eqref{eq:anncyc}}
		\STATE Add $\big( (G_i,B_i,\cM_i), \ann(C^{\circlearrowleft}) \big)$ to $\cA^d$\label{lin:mainoutputcycss2}
		\ENDFOR
		\ENDFOR
		\ENDFOR
		\FOR{$( (G_i,B_i,\cM_i), C) \in \cA^z$}\label{lin:pickcyc}
		\STATE Apply Corollary~\ref{cor:sparsereduction} on $(G_i,B_i,\cM_i)$ and separator $C_\heavy \cup C_\light$\label{lin:mainapplysparse}
		\STATE Add the resulting two subproblems to $res$\label{lin:mainaddtores2}\hfill\algcomment{Uses that $C_\heavy \cup C_\light$ is a sparse separator}
		\ENDFOR
		\RETURN $res$
	\end{algorithmic}
\end{algorithm}
We continue with presenting the clean-up reduction. Our reduction applies a procedure $\texttt{cleanStep}$ (listed in Algorithm~\ref{alg:halvewaste}) for $O(\log |W|)=O(\log k)$ times.
We first present this procedure and its properties.
As mentioned in the beginning of this section, its goal is to reduce the subproblem at hand to subproblems with only few vertices of $W$.
Algorithm~\ref{alg:halvewaste} follows Algorithm~\ref{alg:mainreduction1} closely, and indeed the output subproblems have parameters similar to the subproblems output by Algorithm~\ref{alg:mainreduction1}.

The main difference is how the set $W$ splits.
We guarantee that $|(W\setminus (B_i \setminus B)) \cap V(G_i)| \leq \tfrac{3}{4}|W|$, that is, the number of vertices of $W$ in the subproblem that are not added to the boundary $B$ is at most $\tfrac{3}{4}|W|$.
This gives a handle on the vertices added to the boundary $B$. In particular, we only add $O(k^3 \log k)$ of them and few of them will be pattern vertices: $\upp_{\cM_i}(B_i\setminus B) \leq O(\sqrt{k} \log k)$.

\begin{lemma}[Clean Step Reduction]\label{lem:cleanstep}
	There is a strict reduction that given a subproblem $(G,B,\cM)$ and a set $W \subseteq V(G)$ with $|W|= \Omega(\sqrt{k})$ outputs $l = k^{O(\log k)}$ subproblems $\{(G_i,B_i,\cM_i)\}_{i \leq l}$ with the following properties for every $1 \leq i \leq l$:
	\begin{align*}
		|B_i \setminus B| \leq O(k^3 \log k), \qquad \upp_{\cM_i}(B_i\setminus B) \leq O(\sqrt{k} \log k), \qquad |(W\setminus (B_i \setminus B)) \cap V(G_i)| \leq \tfrac{3}{4}|W|,\\
		|\largeM(\cM_i)| \leq |\largeM(\cM)| + O(\log k), \qquad |\smallM(\cM_i) \setminus \smallM(\cM[V[G_i]])| \leq O(\sqrt{k}\log k).
	\end{align*}
\end{lemma}
\begin{proof}
	The reduction is given in Algorithm~\ref{alg:halvewaste}. Note the algorithm is identical to Algorithm~\ref{alg:mainreduction1} except that it handles weights in a more direct way as they are known in this setting.
	
	Indeed, we start with a cycle $C$ that is $\tfrac{3}{4}$-balanced for $W$ in Line~\ref{lin:balW}. 
	Assume $|W \cap V(\int_G(C))| \geq |W \cap V(\ext_G(C))|$, and that $|V(\ext_{G^\circlearrowleft}(C)) \cap W|\leq |V(\ext_{G^\circlearrowright}(C)) \cap W|$ (the reverse case is symmetric). We see that $|V(\ext_{G^\circlearrowleft}(C)) \cap W| \leq|\sext_{G}(C) \cap W|+|W|/2\leq 3|W|/4$. Thus the cycle picked at Line~\ref{lin:pickcyc} satisfies the same property, and the subproblems at created at Line~\ref{lin:mainaddtores2} indeed satisfy $|(W\setminus V(C)) \cap V(G_i)| \leq \tfrac{3}{4}|W|$.
		
	The remaining part of the correctness is identical to the one in the proof of Lemma~\ref{lem:mainred}.
\end{proof}

Now we present the main clean-up reduction and its properties.
The intuition is that the clean step from Lemma~\ref{lem:cleanstep} is on its own not sufficient to make the subproblem amenable for a main reduction because it does not necessarily decrease $\upp_{\cM}(W)$ to $c\sqrt{k}\log k$ for some constant $c$ (we cannot split the set of pattern vertices in $W$ equally because we don't know them). However, if we apply it $\log |W|$ times on remaining subsets $W$, we ensure only $O(1)$ will remain on top of the vertices added to $B$ during the $O(\log^2 k)$ reductions.

\begin{lemma}[Clean-Up Reduction]\label{lem:shrinkboundred}
	There is a strict reduction that given a subproblem $(G,B,\cM)$ and a set $W \subseteq V(G)$ with $\Omega(\sqrt{k})\leq |W| \leq \poly(k)$ outputs $l = k^{O(\log k)}$ subproblems $\{(G_i,B_i,\cM_i)\}_{i \leq l}$ with the following properties for every $1 \leq i \leq l$:
	\begin{align*}
		|B_i \setminus B| \leq O(k^3 \log^2 k), \qquad
		\upp_{\cM_i}(B_i\setminus B) \leq O(\sqrt{k} \log^2 k), \qquad
		|(W\setminus (B_i \setminus B)) \cap V(G_i)| \leq O(1),\\
	|\largeM(\cM_i)| \leq |\largeM(\cM)| + O(\log^2 k), \qquad
	|\smallM(\cM_i) \setminus \smallM(\cM[V[G_i]])| \leq O(\sqrt{k}\log^2 k).
	\end{align*}
\end{lemma}
\begin{proof}
	The reduction simply applies Algorithm~\ref{alg:halvewaste} $O(\log |W|)$ times (see Algorithm~\ref{alg:cleanup}). By Lemma~\ref{lem:cleanstep} the number of produced subproblems is $(k^{O(\log k)})^{O(\log k)}$, and $|B_i \setminus B|$, $\upp_{\cM_i}(B_i\setminus B)$, $|\largeM(\cM_i)|$ and $|\smallM(\cM_i) \setminus \smallM(\cM[V[G_i]])|$ all satisfy the bound from Lemma~\ref{lem:cleanstep} multiplied with $\log_{4/3}|W|=O(\log k)$.
	
	Since we reduce $W$ with a $3/4$ factor for $\log_{4/3}|W|$ times in each call of $\mathtt{cleanStep}$ the upper bound $|(W\setminus (B_i \setminus B)) \cap V(G_i)| \leq O(1)$ follows.
	
		\begin{algorithm}[H]
		\caption{Clean-Up Reduction}
		\label{alg:cleanup}
		\begin{algorithmic}[1]
			\REQUIRE $\mathtt{cleanUp}(G,B,\cM,W)$\hfill\algcomment{Assumes $W \subseteq V(G)$, $|W|=\poly(k)$ and $G$ is $k$-outerplanar}
			\ENSURE $k^{O(\log^2 k)}$ subproblems with the properties of Lemma~\ref{lem:cleanstep}
			\STATE Let $\cA^0=\{(G,B,\cM)\}$
			\FOR{$d=1,\ldots,z:=\log_{4/3} |W|$}
				\STATE $\cA^d \gets \emptyset$
				\FOR{$(G_i,B_i,\cM_i) \in \cA^{d-1}$}
					\STATE Add all subproblems output by $\mathtt{cleanStep}(G_i,B_i,\cM_i,W \cap V(G_i))$ to $\cA^d$
				\ENDFOR
			\ENDFOR
			\STATE \algorithmicreturn\ $\cA^z$
		\end{algorithmic}
	\end{algorithm}
\end{proof}

Note that if $\mathtt{cleanUp}(G,B,\cM,W)$ is invoked, $B$ is large, and $W$ contains $B$ then the created subproblems will significantly fewer vertices of $B$ by Lemma~\ref{lem:shrinkboundred} because the quantity $|(W\setminus (B_i \setminus B)) \cap V(G_i)|$ that is guaranteed to be constant.

\subsection{The Algorithm: Combining All The Above}
\label{subsec:mainalgo}

Now we use the two reductions from the previous sections recursively to prove Theorem~\ref{thm:main}.
		
\begin{algorithm}[H]
	\caption{Main algorithm to solve a subproblem}
	\label{alg:main}
	\begin{algorithmic}[1]
		\REQUIRE $\mathtt{solveSubproblem}(G,B,\cM)$\hfill\algcomment{Assumes $G$ is $k$-outerplanar}
		\ENSURE Answer of subproblem $(G,B,\cM)$
		\IF{$\min \{ \upp_{\cM}(V(G)), |V(G) \setminus B| \} \leq \sqrt{k}\log^{5} k$}
			\STATE \algorithmicreturn\ answer as computed by Lemma~\ref{lem:basecase}\label{lin:basecase}
		\ELSIF[$c_1 >0$ is a constant chosen later]{$|B| \geq k^4$\textbf{ or }$ |\upp_{\cM_i}(B)| + |\smallM(\cM_i)| \geq c_1\sqrt{k}\log^5 k$}\label{lin:thresholds}
			\STATE $\cA \gets \texttt{cleanUp}(G,B,\cM,B \cup_{M \in \cM}M)$\label{lin:clean}
		\ELSE
			\STATE $\cA \gets \texttt{mainReduce}(G,B,\cM)$\label{lin:main}
		\ENDIF			
		\FOR{$(G_i,B_i,\cM_i) \in \cA$}
			\STATE $a_i \gets \mathtt{solveSubproblem}(G,B,\cM)$
		\ENDFOR
		\STATE \algorithmicreturn\ answer to $(G,B,\cM)$ from reduction of Lemma~\ref{lem:shrinkboundred} and~\ref{lem:mainred} based on answers $\{a_i\}$
	\end{algorithmic}
\end{algorithm}

\begin{proof}[Proof of Theorem~\ref{thm:main}]
	We claim that it suffices to show how to compute $\ind(P,G)$ in the stated time bound: To compute $\sub(P,G)$, let $P'$ and $G'$ be obtained from $P$ and $G$ by replacing every edge $\{x,y\}$ with triangle on vertices $u_{xy},v_{xy},w_{xy}$ and with edges $\{u_{xy},x\}$ and $\{u_{x,y},v\}$. It is easy to see that $\sub(P,G)=|\ind(P',G')|/2^{|E(P)|}$, and that $|\sepsa{\tO(\sqrt{k}),P'}|$ is at most $|\sepsa{\tO(\sqrt{k}),P}|^c$ for some $c$.
	
	We use Algorithm~\ref{alg:main} to solve the subproblem $(G,\emptyset,\{(V(G),0,k)\})$;
	note we can indeed assume $G$ to be $k$-outerplanar by Lemma~\ref{lem:redkout}. 
	The correctness of the algorithm follows directly from Lemma~\ref{lem:shrinkboundred} and~\ref{lem:mainred}.
	Note that at Line~\ref{lin:basecase} we can assume $\upp_{\cM}(V(G)) \leq O(\sqrt{k}\log^5 k)$ since $\upp_{\cM}(V(G)) \leq O(\sqrt{k}\log^5 k)$.

	\paragraph{Running Time}
	We first focus on the number of subproblems generated throughout the algorithm.
	We assign a potential $\varphi$ to a subproblem that estimates how soon a cleaning step occurs, defined as follows:
	\[
		\varphi((G,B,\cM))= \frac{|B|}{k^4}+\frac{\upp(B)+|\smallM(\cM)|}{c_1\sqrt{k}\log^6 k},
	\]
	for some large enough constant $c_1$.
	Note that if a cleaning step occurs, then $\varphi(G,B,\cM)\geq 1$, and afterwards $\varphi \leq 1/(c_1\log^3 k)+1/k$, which is at most $1/\log^{-3}k$ for large enough $c_1$ and $k>1$.	
	We define $T(\nu,u,\varphi)$ as the number of subproblems generated by $\mathtt{solveSubproblem}(G,B,\cM)$ where $|V(G)\setminus B|=\nu$, $\upp_{\cM}(B)=u$ and $\varphi=\varphi(G,B,\cM)$. By Lemma's~\ref{lem:shrinkboundred} and~\ref{lem:mainred} we have:
	\[
		T(\nu,u,\varphi) \leq
		\begin{cases}
		1,& \text{if Line~\ref{lin:basecase} is reached }\\
		k^{O( \log k)} T(\nu,u,\log^{-3} k), & \text{if Line~\ref{lin:clean} is reached} \\
		 2T\big(\tfrac{3}{4}\nu,u,\varphi+\log^{-3} k\big)+ k^{O(\log k)}T\Big(\nu,u-\Omega\big(\sqrt{k}\big),\varphi+\log^{-3} k\Big),& \text{if Line~\ref{lin:main} is reached}.
		\end{cases}
	\]

	We claim that
	\[
		T(\nu,u,\varphi) \leq \nu^4 \exp\left(c_2\left(\frac{u}{\sqrt{k}}+\varphi\right)\log^{2} k\right),
	\]
	for some large enough constant $c_2$.
	To see this, first note the base case where Line~\ref{lin:basecase} is reached is trivial. Moreover, if Line~\ref{lin:clean} is reached, we have for some constant $c_3$ that,
	\begin{align*}
		T(\nu,u,\varphi) &\leq k^{O(\log k)}\nu^4\exp\left(c_2\left(\frac{u}{\sqrt{k}}+\log^{-3}k\right)\log^{2} k\right)\\
		&\leq \nu^4\exp\left(c_2\left(\frac{u}{\sqrt{k}}+\log^{-3}k+c_3/c_2\right)\log^{2} k\right)\\
		&\leq \nu^4\exp\left(c_2\left(\frac{u}{\sqrt{k}}+1\right)\log^{2} k\right),
	\end{align*}
	where the last inequality assumes $c_2$ is chosen to be at least twice $c_3$.
	
	Finally, if Line~\ref{lin:main} is reached then $T(\nu,u,\varphi)=A+B$ where
	\begin{align*}
		 A&\leq 2(\tfrac{3}{4}\nu)^4\exp\left(c_2\left(\frac{u}{\sqrt{k}}+\varphi+\log^{-3}k\right)\log^{2} k\right)\\
		 &\leq 2(\tfrac{3}{4})^4\exp(c_2\log^{-1}k) \nu^4\exp\left(c_2\left(\frac{u}{\sqrt{k}}+\varphi\right)\log^{2}k\right)\\
		 &\leq \tfrac{2}{3}\nu^4 \exp\left(c_2\left(\frac{u}{\sqrt{k}}+\varphi\right)\log^{2} k\right),
	\end{align*}
	where the last inequality uses $c_2$ is chosen large enough,
	and for some constants $c_4,c_5$ we have
	\begin{align*}
		B &= k^{c_4 \log k}\exp\left(c_2\left(\frac{u-c_5\sqrt{k}}{\sqrt{k}}+\varphi+\log^{-3} k\right)\log^{2} k\right)\\
		&\leq \exp\left(c_2\left(\frac{u}{\sqrt{k}}-c_5+c_4/c_2+\varphi+\log^{-2} k\right)\log^{2} k\right)\\
		&\leq \tfrac{1}{3}\nu^4 \exp\left(c_2\left(\frac{u}{\sqrt{k}}+\varphi\right)\log^{2} k\right),
	\end{align*}
	where the last inequality uses that $c_2$ is chosen large enough and that $k$ is large enough. Thus the claim holds in this case as well.
	
	Now we focus on the time spent per subproblem. Since algorithm Lemma~\ref{lem:basecase} runs in the claimed time bound $2^{\tilde{O}(\sqrt{k})}|\sepsa{\tO(\sqrt{k})}|n^{O(1)}$ and both reductions $\texttt{cleanUp}$ and $\texttt{mainReduce}$ are strict, it suffices that the sizes of the input and outputs to all subproblems is at most $2^{\tilde{O}(\sqrt{k})}|\sepsa{\tO(\sqrt{k})}|n^{O(1)}$.
	
	To see that this is true, note $|B| = O(k^4)$ and $|\upp_{\cM_i}(B)|,|\smallM(\cM_i)|\leq O(\sqrt{k}\log^5 k)$ due to the check and cleaning step at Lines~\ref{lin:thresholds} and~\ref{lin:clean}. Moreover we claim that, $|\largeM(\cM)| \leq O(\sqrt{k}\log^3k + \log n / \log^4 k)$ for any generated subproblem $(G,B,\cM)$.
	To see this, first note that during a recursive call the procedure $\texttt{mainReduce}$ is invoked at most $O(\sqrt{k}+\log n)$ times as in each call either $\upp_{\cM}(V(G))$ is decreased with $\Omega(k)$ or $V(G) \setminus B$ is decreased with a constant factor.
	The number of large monitors added by $\texttt{mainReduce}$ is at most $O(\sqrt{k} \log k)$ since it only add large monitors to generated subproblems with $\upp_{\cM}(V(G))$ being decreased with $\Omega(k)$.
	
	Second, the number of call to $\texttt{cleanUp}$ is at most $O(\log^{-4} k)$ times the number of calls to $\texttt{mainreduce}$ as this number of calls is needed to increase $|\upp_{\cM}(B)|$ or $|\smallM(\cM)|$ to be large enough so the condition at Line~\ref{lin:thresholds} holds.
	As in each call to $\texttt{cleanUp}$ only $O(\log^2 k)$ large monitors are added, we add a total of $O((\sqrt{k}+\log n)\log^{-2}k)$, and hence the number of possibilities $k^{|\largeM(\cM)|}$ for counter of large monitors is at most $2^{O(\sqrt{k})}n^{O(1)}$, as required.
\end{proof}

\paragraph{Acknowledgements}The author thanks Radu Curticapean, Viresh Patel and Guus Regts for discussions.
In particular Viresh Patel and Guus Regts for an inspiring explanation of~\cite{patelregts}, and Viresh Patel for posing Question 8.3 from~\cite{patelregts} during a `training week' of the NETWORKS project (which was the starting point of this research).
Moreover, Radu Curticapean pointed out the $O^*(2^{O(k)})$ time algorithm from~\cite{DBLP:conf/stacs/Dorn10} and enriched the author's vocabulary with the verb `to include-exclude'. 

\bibliographystyle{alpha}
\bibliography{main}

\newcommand{\etalchar}[1]{$^{#1}$}
\begin{thebibliography}{MSOI{\etalchar{+}}02}

\bibitem[AYZ95]{DBLP:journals/jacm/AlonYZ95}
Noga Alon, Raphael Yuster, and Uri Zwick.
\newblock Color-coding.
\newblock {\em J. {ACM}}, 42(4):844--856, 1995.

\bibitem[Bak94]{DBLP:journals/jacm/Baker94}
Brenda~S. Baker.
\newblock Approximation algorithms for {NP}-complete problems on planar graphs.
\newblock {\em J. {ACM}}, 41(1):153--180, 1994.

\bibitem[Bie15]{BIEDL2015275}
Therese Biedl.
\newblock On triangulating k-outerplanar graphs.
\newblock {\em Discrete Applied Mathematics}, 181:275 -- 279, 2015.

\bibitem[BNvdZ16]{DBLP:conf/icalp/BodlaenderNZ16}
Hans~L. Bodlaender, Jesper Nederlof, and Tom~C. van~der Zanden.
\newblock Subexponential time algorithms for embedding h-minor free graphs.
\newblock In {\em 43rd International Colloquium on Automata, Languages, and
  Programming, {ICALP} 2016, July 11-15, 2016, Rome, Italy}, pages 9:1--9:14,
  2016.

\bibitem[Bod98]{DBLP:journals/tcs/Bodlaender98}
Hans~L. Bodlaender.
\newblock A partial \emph{k}-arboretum of graphs with bounded treewidth.
\newblock {\em Theor. Comput. Sci.}, 209(1-2):1--45, 1998.

\bibitem[CDM17]{DBLP:conf/stoc/CurticapeanDM17}
Radu Curticapean, Holger Dell, and D{\'{a}}niel Marx.
\newblock Homomorphisms are a good basis for counting small subgraphs.
\newblock In Hamed Hatami, Pierre McKenzie, and Valerie King, editors, {\em
  Proceedings of the 49th Annual {ACM} {SIGACT} Symposium on Theory of
  Computing, {STOC} 2017, Montreal, QC, Canada, June 19-23, 2017}, pages
  210--223. {ACM}, 2017.

\bibitem[CFHW17]{cygan_et_al:DR:2017:7247}
Marek Cygan, Fedor~V. Fomin, Danny Hermelin, and Magnus Wahlstr{\"o}m.
\newblock {Randomization in Parameterized Complexity (Dagstuhl Seminar 17041)}.
\newblock {\em Dagstuhl Reports}, 7(1):103--128, 2017.

\bibitem[CFK{\etalchar{+}}15]{Cygan:2015:PA:2815661}
Marek Cygan, Fedor~V. Fomin, Lukasz Kowalik, Daniel Lokshtanov, Daniel Marx,
  Marcin Pilipczuk, Michal Pilipczuk, and Saket Saurabh.
\newblock {\em Parameterized Algorithms}.
\newblock Springer Publishing Company, Incorporated, 1st edition, 2015.

\bibitem[CR06]{10.1007/11830924_10}
Shuchi Chawla and Tim Roughgarden.
\newblock Single-source stochastic routing.
\newblock In Josep D{\'i}az, Klaus Jansen, Jos{\'e} D.~P. Rolim, and Uri Zwick,
  editors, {\em Approximation, Randomization, and Combinatorial Optimization.
  Algorithms and Techniques}, pages 82--94, Berlin, Heidelberg, 2006. Springer
  Berlin Heidelberg.

\bibitem[Cur16]{DBLP:conf/esa/Curticapean16}
Radu Curticapean.
\newblock Counting matchings with k unmatched vertices in planar graphs.
\newblock In Piotr Sankowski and Christos~D. Zaroliagis, editors, {\em 24th
  Annual European Symposium on Algorithms, {ESA} 2016, August 22-24, 2016,
  Aarhus, Denmark}, volume~57 of {\em LIPIcs}, pages 33:1--33:17. Schloss
  Dagstuhl - Leibniz-Zentrum fuer Informatik, 2016.

\bibitem[Cur18]{DBLP:conf/iwpec/Curticapean18}
Radu Curticapean.
\newblock Counting problems in parameterized complexity.
\newblock In Christophe Paul and Michal Pilipczuk, editors, {\em 13th
  International Symposium on Parameterized and Exact Computation, {IPEC} 2018,
  August 20-24, 2018, Helsinki, Finland}, volume 115 of {\em LIPIcs}, pages
  1:1--1:18. Schloss Dagstuhl - Leibniz-Zentrum fuer Informatik, 2018.

\bibitem[CW98]{Calkin:1998:NIS:288856.288860}
Neil~J. Calkin and Herbert~S. Wilf.
\newblock The number of independent sets in a grid graph.
\newblock {\em SIAM J. Discret. Math.}, 11(1):54--60, February 1998.

\bibitem[DFHT05]{DBLP:journals/jacm/DemaineFHT05}
Erik~D. Demaine, Fedor~V. Fomin, Mohammad~Taghi Hajiaghayi, and Dimitrios~M.
  Thilikos.
\newblock Subexponential parameterized algorithms on bounded-genus graphs and
  \emph{H}-minor-free graphs.
\newblock {\em J. {ACM}}, 52(6):866--893, 2005.

\bibitem[DFL{\etalchar{+}}13]{DBLP:journals/iandc/DornFLRS13}
Frederic Dorn, Fedor~V. Fomin, Daniel Lokshtanov, Venkatesh Raman, and Saket
  Saurabh.
\newblock Beyond bidimensionality: Parameterized subexponential algorithms on
  directed graphs.
\newblock {\em Inf. Comput.}, 233:60--70, 2013.

\bibitem[DLN{\etalchar{+}}09]{DBLP:conf/coco/DattaLNTW09}
Samir Datta, Nutan Limaye, Prajakta Nimbhorkar, Thomas Thierauf, and Fabian
  Wagner.
\newblock Planar graph isomorphism is in log-space.
\newblock In {\em Proceedings of the 24th Annual {IEEE} Conference on
  Computational Complexity, {CCC} 2009, Paris, France, 15-18 July 2009}, pages
  203--214. {IEEE} Computer Society, 2009.

\bibitem[Dor07]{DBLP:conf/wg/Dorn07}
Frederic Dorn.
\newblock How to use planarity efficiently: New tree-decomposition based
  algorithms.
\newblock In Andreas Brandst{\"{a}}dt, Dieter Kratsch, and Haiko M{\"{u}}ller,
  editors, {\em Graph-Theoretic Concepts in Computer Science, 33rd
  International Workshop, {WG} 2007, Dornburg, Germany, June 21-23, 2007.
  Revised Papers}, volume 4769 of {\em Lecture Notes in Computer Science},
  pages 280--291. Springer, 2007.

\bibitem[Dor10]{DBLP:conf/stacs/Dorn10}
Frederic Dorn.
\newblock Planar subgraph isomorphism revisited.
\newblock In {\em 27th International Symposium on Theoretical Aspects of
  Computer Science, {STACS} 2010, March 4-6, 2010, Nancy, France}, pages
  263--274, 2010.

\bibitem[EGT]{laudatio}
David Eppstein, Georg Gottlob, and Jan~Arne Telle.
\newblock Nerode prize 2015: Winner announcement.
\newblock Available at http://eatcs.org/images/awards/Nerode15-laudation.pdf.

\bibitem[EKMM16]{erickson_et_al:DR:2016:6722}
Jeff Erickson, Philip~N. Klein, D{\'a}niel Marx, and Claire Mathieu.
\newblock {Algorithms for Optimization Problems in Planar Graphs (Dagstuhl
  Seminar 16221)}.
\newblock {\em Dagstuhl Reports}, 6(5):94--116, 2016.

\bibitem[Epp99]{DBLP:journals/jgaa/Eppstein99}
David Eppstein.
\newblock Subgraph isomorphism in planar graphs and related problems.
\newblock {\em J. Graph Algorithms Appl.}, 3(3), 1999.

\bibitem[FLM{\etalchar{+}}16]{pattern}
Fedor~V. Fomin, Daniel Lokshtanov, D{\'{a}}niel Marx, Marcin Pilipczuk, Michal
  Pilipczuk, and Saket Saurabh.
\newblock Subexponential parameterized algorithms for planar and
  apex-minor-free graphs via low treewidth pattern covering.
\newblock In Irit Dinur, editor, {\em {IEEE} 57th Annual Symposium on
  Foundations of Computer Science, {FOCS} 2016, 9-11 October 2016, Hyatt
  Regency, New Brunswick, New Jersey, {USA}}, pages 515--524. {IEEE} Computer
  Society, 2016.

\bibitem[Fri04]{DBLP:journals/mst/Frick04}
Markus Frick.
\newblock Generalized model-checking over locally tree-decomposable classes.
\newblock {\em Theory Comput. Syst.}, 37(1):157--191, 2004.

\bibitem[GJM15]{DBLP:journals/jcss/GoldbergJM15}
Leslie~Ann Goldberg, Mark Jerrum, and Colin McQuillan.
\newblock Approximating the partition function of planar two-state spin
  systems.
\newblock {\em J. Comput. Syst. Sci.}, 81(1):330--358, 2015.

\bibitem[Gri11]{DBLP:journals/dmtcs/Grigoriev11}
Alexander Grigoriev.
\newblock Tree-width and large grid minors in planar graphs.
\newblock {\em Discrete Mathematics {\&} Theoretical Computer Science},
  13(1):13--20, 2011.

\bibitem[HW74]{Hopcroft:1974:LTA:800119.803896}
John~E. Hopcroft and Jin~K. Wong.
\newblock Linear time algorithm for isomorphism of planar graphs (preliminary
  report).
\newblock In {\em Proceedings of the Sixth Annual ACM Symposium on Theory of
  Computing}, STOC '74, pages 172--184, New York, NY, USA, 1974. ACM.

\bibitem[Kas61]{KASTELEYN19611209}
Pieter~W. Kasteleyn.
\newblock The statistics of dimers on a lattice: I. the number of dimer
  arrangements on a quadratic lattice.
\newblock {\em Physica}, 27(12):1209 -- 1225, 1961.

\bibitem[KM]{planarity}
Philip Klein and Shay Mozes.
\newblock Planarity (book draft).
\newblock \url{http://planarity.org/}.

\bibitem[LT79]{doi:10.1137/0136016}
Richard~J. Lipton and Robert~Endre Tarjan.
\newblock A separator theorem for planar graphs.
\newblock {\em SIAM Journal on Applied Mathematics}, 36(2):177--189, 1979.

\bibitem[MP15]{DBLP:conf/mfcs/MontanariP15}
Sandro Montanari and Paolo Penna.
\newblock On sampling simple paths in planar graphs according to their lengths.
\newblock In Giuseppe~F. Italiano, Giovanni Pighizzini, and Donald Sannella,
  editors, {\em Mathematical Foundations of Computer Science 2015 - 40th
  International Symposium, {MFCS} 2015, Milan, Italy, August 24-28, 2015,
  Proceedings, Part {II}}, volume 9235 of {\em Lecture Notes in Computer
  Science}, pages 493--504. Springer, 2015.

\bibitem[MSOI{\etalchar{+}}02]{Milo824}
R.~Milo, S.~Shen-Orr, S.~Itzkovitz, N.~Kashtan, D.~Chklovskii, and U.~Alon.
\newblock Network motifs: Simple building blocks of complex networks.
\newblock {\em Science}, 298(5594):824--827, 2002.

\bibitem[MT92]{DBLP:journals/dm/MatousekT92}
Ji{\v{r}}{\'{i}} Matou{\v{s}}ek and Robin Thomas.
\newblock On the complexity of finding iso- and other morphisms for partial
  k-trees.
\newblock {\em Discrete Mathematics}, 108(1-3):343--364, 1992.

\bibitem[PR17a]{DBLP:journals/corr/PatelR17}
Viresh Patel and Guus Regts.
\newblock Computing the number of induced copies of a fixed graph in a bounded
  degree graph.
\newblock {\em CoRR}, abs/1707.05186, 2017.

\bibitem[PR17b]{patelregts}
Viresh Patel and Guus Regts.
\newblock Deterministic polynomial-time approximation algorithms for partition
  functions and graph polynomials.
\newblock {\em SIAM Journal on Computing}, 46(6):1893--1919, 2017.

\bibitem[Sch09]{thesisSchweitzer}
Pascal Schweitzer.
\newblock Problems of unknown complexity : graph isomorphism and ramsey
  theoretic numbers.
\newblock 09 2009.

\bibitem[Taz12]{DBLP:journals/tcs/Tazari12}
Siamak Tazari.
\newblock Faster approximation schemes and parameterized algorithms on
  (odd-){H}-minor-free graphs.
\newblock {\em Theor. Comput. Sci.}, 417:95--107, 2012.

\bibitem[YZ13]{Yin:2013:ACV:2627817.2627821}
Yitong Yin and Chihao Zhang.
\newblock Approximate counting via correlation decay on planar graphs.
\newblock In {\em Proceedings of the Twenty-fourth Annual ACM-SIAM Symposium on
  Discrete Algorithms}, SODA '13, pages 47--66, Philadelphia, PA, USA, 2013.
  Society for Industrial and Applied Mathematics.

\end{thebibliography}



\end{document}